\documentclass[pdflatex,sn-mathphys-num,a4paper]{sn-jnl}
\usepackage{graphicx}%
\usepackage{multirow}%
\usepackage{amsmath,amssymb,amsfonts}%
\usepackage{amsthm}%
\usepackage{mathrsfs}%
\usepackage[title]{appendix}%
\usepackage{xcolor}%
\usepackage{textcomp}%
\usepackage{manyfoot}%
\usepackage{booktabs}%
\usepackage{algorithm}%
\usepackage{algorithmicx}%
\usepackage{algpseudocode}%
\usepackage{listings}%
\usepackage{subcaption}
\theoremstyle{thmstyleone}%
\newtheorem{theorem}{Theorem}%  meant for continuous numbers
\newtheorem{lemma}[theorem]{Lemma}
\newtheorem{corollary}[theorem]{Corollary}

\theoremstyle{thmstyletwo}%

\theoremstyle{thmstylethree}%
\newtheorem{definition}{Definition}%

\newcommand{\highlight}[1]{{#1}}

\begin{document}

\title[Article Title]{The uncloneable bit exists}

\author*[1]{\fnm{Archishna} \sur{Bhattacharyya}}\email{abhat086@uottawa.ca}
\equalcont{These authors contributed equally to this work.}

\author*[1]{\fnm{Anne} \sur{Broadbent}}\email{abroadbe@uottawa.ca}

\author*[2,3]{\fnm{Eric} \sur{Culf}}\email{eculf@uwaterloo.ca}
\equalcont{These authors contributed equally to this work.}

\affil[1]{\orgname{Faculty of Science, University of Ottawa}, \orgaddress{\street{30 Marie Curie}, \city{Ottawa}, \\\postcode{K1N~6N5}, \state{ON}, \country{Canada}}}

\affil[2]{\orgdiv{Institute for Quantum Computing}, \orgname{University of Waterloo}, \orgaddress{\street{\\200 University Ave W}, \city{Waterloo}, \postcode{N2L 3G1}, \state{ON}, \country{Canada}}}

\affil[3]{\orgname{Perimeter Institute for Theoretical Physics}, \orgaddress{\street{31 Caroline St N}, \city{Waterloo}, \postcode{N2L 2Y5}, \state{ON}, \country{Canada}}}

\abstract{We establish quantum uncloneable encryption with unconditional security, preventing two non‑communicating adversaries from simultaneously decrypting a single ciphertext — even when both are given the key. Our construction achieves security that approaches the ideal limit at a rate that is exponentially small in the security parameter, without employing any assumptions. \highlight{Our proof invokes unitary invariance of the shared entangled state and simplifies the adversarial strategies by enforcing this symmetry. Crucially, it then rules out the sender being highly correlated with two non‑communicating adversaries at once by an approximation property that we develop, for such unitarily invariant states, which yields a near-optimal bound on the probability of cloning}. Consequently, no coordinated strategy beats random guessing of the encrypted bit, establishing unconditional uncloneability. This reveals the existence of an uncloneable bit in Nature and delineates a fundamental, physically enforced cryptographic primitive unavailable in classical settings.}

\vspace{2cm}




\maketitle

\newpage

The idea of no-cloning (Fig.~\ref{fig:bit}) has been the subject of great intrigue in modern science~\cite{Die82, WZ82}. Possibly, the most coveted application of this idea is the proposal of an \textit{uncloneable} bit. This is because it enables the notion of classically-impossible security, that is stronger, and in a sense, a reincarnation of no-cloning in quantum physics. The no-cloning principle when used to achieve classically-impossible security guarantees defines the paradigm of uncloneable cryptography, underpinning much of the groundbreaking work in quantum cryptography, notably quantum key distribution~\cite{BB84} and quantum money~\cite{Wie83}. It may be noted that approximate cloning is possible, but even optimal cloners for general states incur high error~\cite{Wer98}, hence, it does not prohibit uncloneable cryptography. An uncloneable bit is a single bit of information encoded into a quantum state serving as a ciphertext, that precludes cloning of the encoded bit by two cooperating but non-interacting adversaries. Thus, as information that cannot be copied, it is the idea of an uncloneable encryption scheme~\cite{BL20}. Over the years, the assumption of the feasibility of the uncloneable bit has enabled several uncloneable cryptographic primitives such as quantum copy-protection~\cite{CLLZ21}, secure software leasing~\cite{ALP21}, certified deletion~\cite{BI20}, uncloneable decryption~\cite{GZ20eprint}, uncloneable zero-knowledge proofs~\cite{JK25}, and many more. Yet, a proof of security establishing the existence of an uncloneable bit remained elusive, even under computational assumptions.

\begin{figure}[b]
    \centering
    \includegraphics[width=1\textwidth]{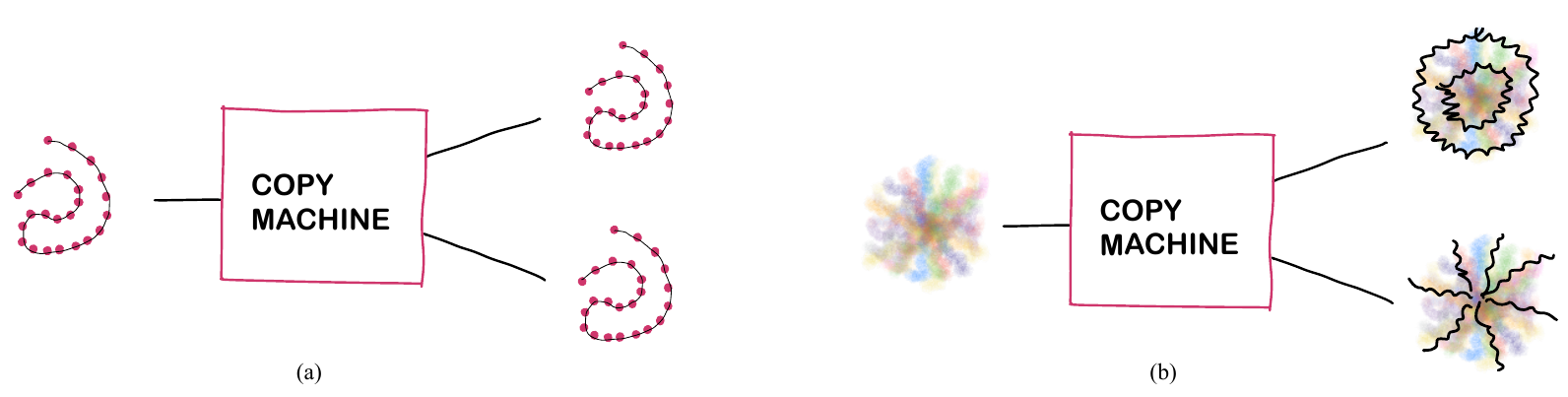}
    \caption{\textbf{Classical and quantum cloning.} (a) shows a classical bit can be perfectly copied. (b) depicts that it is impossible to clone a quantum bit perfectly. The quantum bit undergoing a copying operation similar to its classical counterpart, outputs two residual objects that are no longer true copies of the original input. This highlights a fundamental distinction between classical and quantum information.}
    \label{fig:bit}
\end{figure}

Alternative approaches were considered, such as a heuristic proof of security in the quantum random oracle model~\cite{BL20,AKL+22}. A variety of candidate schemes were proposed, but their security could neither be proved without any assumptions, nor with computational assumptions that are well-justified via existing constructions, such as one-way functions. For example, \cite{CHV24arxiv}~presented candidates relying on conjectures about uncloneable forms of indistinguishability obfuscation, and \cite{BBCNPR24}~considered a relaxed notion of security that deviates from the negligible scaling proposed originally. Other work concentrated on variants of uncloneable encryption, notably with interaction~\cite{BC23arxiv}, or with quantum keys~\cite{AKY24arxiv}. Some even considered the possibility that an uncloneable bit cannot exist, proving no-go theorems on possible schemes~\cite{MST21, AKL+22, CKNY25}. Alas, the pursuit of the `uncloneable bit' seemed to question our everyday intuition of whether there is any operational meaning associated with encoding information. 

In particular, it was determined from the no-go theorems that the states in any such encoding must be highly mixed. Recently, \cite{BC26}~used the insight that \textit{decoupling} should inherently be tied to this fact. It is the idea that in a tripartite quantum system, two parties are said to `decouple' if they are separable, implying a notion of statistical independence between them. \cite{BC26}~proved security of a scheme showing security $\widetilde{O}\left(\tfrac{1}{q}\right)$ in the number of qubits $q$ with no computational assumptions, achieving the relaxed notion of \textit{weak uncloneable security} introduced in~\cite{BBCNPR24}. \highlight{The main idea was that a sufficient condition for decoupling guarantees randomness extraction on $A$ in a fully quantum generalisation of \textit{privacy amplification} \cite{TSSR11} which was the main obstacle to proving security, and hence existence of uncloneable encryption.} The work left open the question of achieving the full strength of the uncloneable bit, which requires a negligible bound~\cite{BL20}, and is known as \textit{strong uncloneable security}. Here, we report the latter, showing that there exists a scheme with security $\exp(-q)$ in the number of qubits~$q$, without any computational assumptions; unveiling the uncloneable bit in Nature. 

In hindsight, it would be rather surprising if the conjectured existence of the uncloneable bit turned out to be false. It would imply that contrary to our physical intuition, quantum information is `more' classical than we expected, and perhaps wished for! The existence of the uncloneable bit restores order, reinstating that quantum information is uncloneable, in a stronger sense after all. One may note that the uncloneable bit is a fundamental primitive of uncloneable cryptography, and due to~\cite{HKNY24}, it can be used to construct secure uncloneable encryption schemes for messages of arbitrary length, under standard cryptographic assumptions (see Section~\ref{sec:UE}). The uncloneable bit is characterised as a family of \textit{quantum encryptions of classical messages} encoding a single bit that can be scaled to have arbitrarily good uncloneable security, measured by the probability with which a cloning attack succeeds on the encryption. This notion of security also admits a dual characterisation as the winning probability of a \textit{monogamy-of-entanglement} game in which the adversaries try to jointly guess the encoded bit by making measurements on a shared entangled state. \textit{Monogamy} is a property of quantum entanglement~\cite{Terhal04} which asserts that in a tripartite system, $\phi_{ABC}$, if $B$ is highly entangled with $A$, then $C$ can only be weakly entangled with $A$. One of the ways in which the strength of this property is quantified is via the winning probability of such a game, first defined in~\cite{TFKW13}. The security of the uncloneable bit modelled by such a game is depicted in Fig.~\ref{fig:encode-game}.  

\begin{figure}
    \centering
    \includegraphics[width=0.9\textwidth]{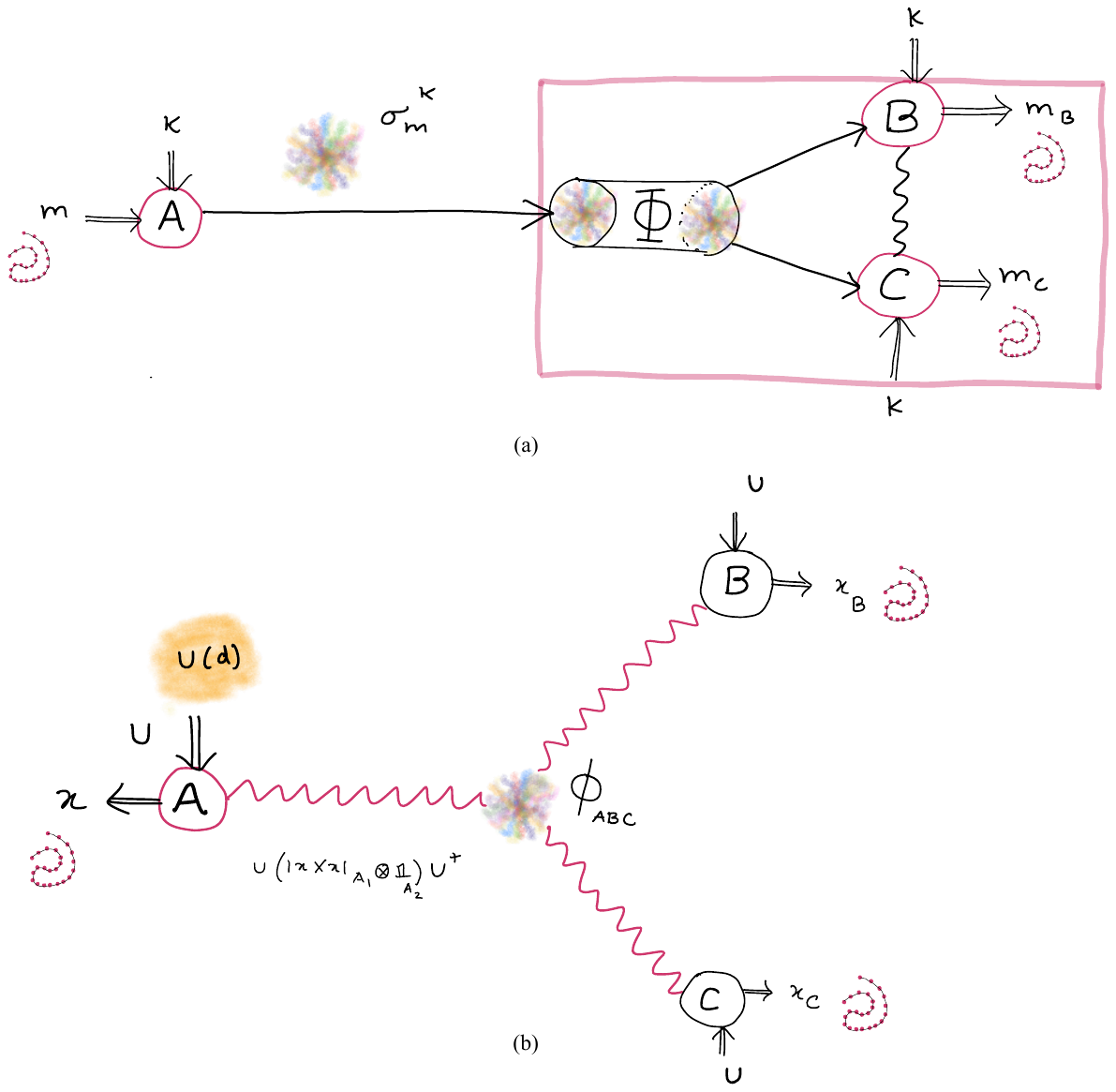}
    \caption{\textbf{Uncloneable encryption.} (a) depicts an uncloneable encryption scheme abstractly. $A$ samples a random key $k$ and a bit message $m \in \{0, 1\}$ uniformly and prepares the ciphertext $\sigma_m^k$, which is then subjected to a cloning attack represented by the box: $\Phi$ is an adversarially-chosen pirate channel through which $\sigma_m^k$ passes, outputting an entangled state in the system of the adversaries, $B$ and $C$. $A$ informs $B$ and $C$ of the key, $k$. Without communicating, $B$ and $C$ both try to guess $m.$ If $m_B = m_C = m$, the cloning attack succeeds, and the cloning probability is said to be high. This is known as the prepare-and-measure picture. (b) shows a picture dual to (a) that is entanglement-based, and is described by a monogamy-of-entanglement game. $B$ and $C$ prepare an entangled state $\phi_{ABC}$ (the Choi state of the pirate channel, $\Phi$ in the dual picture) and share it with the honest referee, $A$ after which they can no longer communicate. $A$ samples a question \highlight{$U$} and informs $B$ and $C$. $A$ makes a measurement specified by \highlight{$U$} to get answer $x.$ Then, $B$ and $C$ measure their parts of the state and attempt to guess $x$. The players $B$ and $C$ win if both their guesses are correct, i.e., $x_B = x_C = x.$ The scheme we show security for is \highlight{the Haar-measure encryption, and is dual to the Haar-measure game. Here, \highlight{$U$} depicts a Haar-random unitary.} The uncloneable security we show is that, in the limit of large dimension, the success probability of a cloning attack on the \highlight{Haar encryption} must tend to $\frac{1}{2}$, exponentially in the size of the encoding \highlight{($\exp(-q)$ in $q$, the number of qubits)}. This is equivalent to showing that $B$ and $C$ cannot do better than a coordinated random guess of the encoded bit in the \highlight{Haar-measure} game.}
    \label{fig:encode-game}
\end{figure}

We show that uncloneable encryption is unconditionally secure. We first define the scheme with which we work.

\highlight{\begin{definition}[Scheme]\label{def:haar-qecm}
    Let $d\in\mathbb{N}$ be even, let $A_1=\{0,1\}$, and $A_2=[d/2]$. Set $A=A_1A_2$. Let $\sigma_0=\frac{2}{d}|0\rangle\!\langle 0|\otimes I\in D(\mathcal{H}_A)$ and $\sigma_1=\frac{2}{d}|1\rangle\!\langle 1|\otimes I\in D(\mathcal{H}_A)$. The \emph{$d$-dimensional Haar-measure encryption of a bit} is the QECM $\texttt{Q}_{\mathcal{U}(d),2}=(\mathcal{U}(\mathcal{H}_A),\{0,1\},A,\mu_{\mathcal{U}(\mathcal{H}_A)},\{U\sigma_xU^\dag\}_{U\in\mathcal{U}(\mathcal{H}_A),x\in\{0,1\}})$.
\end{definition}

The security we achieve is stated as follows which holds due to the bound obtained in Theorem \ref{thm:yay}. One may observe that the security bound on the Haar-measure encryption obtained in Theorem \ref{thm:yay} is near-optimal, $\frac{1}{2}+O(\frac{1}{d^{1/8}})$. The best-known lower bound is $\frac{1}{2}+O(\frac{1}{d^{1/2}})$ as reported in~\cite{BCR25}.

\begin{theorem}[Security] \label{thm:yay-intro}
    The success probability of any cloning attack on $\texttt{Q}_{\mathcal{U}(2^n),2}$ approaches~$\frac{1}{2}$ exponentially in $n$.
\end{theorem}

Compared to previous works showing progress on this problem including \cite{BC26}, we introduce a novel approach to proving security for uncloneable encryption, and simultaneously obtain full security and a near-optimal bound, surpassing any known bound to the best of our knowledge. Although interestingly, the scheme we show security for is the same as that in \cite{BC26}. One crucial difference is that in our case, the scheme no longer admits an efficient construction. In Section \ref{sec-disc} we elaborate on the possible reason for this. 

The first step in our proof is to notice that the pirate channel $\Phi$ in the encryption scheme is unitarily equivariant, as a result of which in the dual, entanglement-based picture, the input state to the protocol, $\phi_{ABC}$ satisfies unitary invariance. We exploit this symmetry to simplify the adversarial strategies that $B$ and $C$ can apply in a cloning attack. Our analysis begins by establishing unitary invariance on arbitrary cloning attacks without loss of generality (see Section \ref{sec:u-inv}). Once this is accomplished, our next step is to treat the adversarial measurements which are polynomials on the unitary group of the unitary sampled by $A$ in the Haar-measure game. We perform an analysis to ensure the degree of any such polynomial is sufficiently low such that we now obtain bounded-order unitarily invariant strategies (see Section \ref{sec:bdd-u-inv}). 

One may observe that throughout the analysis, the two adversaries are treated as acting simultaneously at all times rather than individually. The purpose of these two steps is to create a paradigm in which such adversarial behaviour, without loss of generality, can be controlled in favour of the honest party, $A$. This now sets the stage to apply our main technical invention. We establish an approximation property of any input state $\phi_{ABC}$ that is arbitrary beyond unitary invariance to a typical state in the space of such unitarily invariant states, such that any information-processing claim on this typical state will enforce the same for all such states in the unitarily-invariant space. Thus, a claim establishing security by this approximation property will also follow suit. This approximation theorem mimicks a \textit{half de Finetti theorem} that guarantees a de Finetti-style representation for $n$-exchangeable unitarily-invariant states \cite{CKMR07}. However, the main difference between that and our case is that we have no such notion of $n$-exchangeability or permutation invariance in our input states. Hence, we mark the distinction by not referring to it as a `half' de Finetti representation, but also note that such an analogy provides a more lucid conceptual understanding of our approximation property presented in Section \ref{sec:half-deFin}. 

Interestingly, the maximally entangled state is invariant under conjugation by unitary representations, and the most general attack state is a convex combination of maximal entanglement between $A$ and $B$, and no entanglement between $A$ and $B$ which implies, in a purified larger space, there exists maximal entanglement between $A$ and $C$ in the latter part. We refer to these general attack states as \textit{devious} states for better or for worse. The devious states do not correspond to a counterexample on any scheme, nor do they represent an explicit attack, however, they prohibitively interfere with randomness extraction by techniques that use decoupling to decorrelate either of the adversaries $B$ or $C$, and further rest on the symmetry between them to argue security. Hence, to handle these devious states, we formalise a framework where adversarial strategies of $B$ and $C$ are treated simultaneously, and any argument using the exchange symmetry between $B$ and $C$ in the input $\phi_{ABC}$ is avoided. However, the stronger unitary invariance property of the input states proves helpful in analysing security as discussed above. It is to our intrigue that the space of such unitarily invariant attack states is spanned by a generalisation of these devious states. Thus, in some sense our approximation property reduces a security claim for arbitrary unitarily invariant states to the strongest possible attack states which also happen to be symmetric in the unitarily invariant space, reflecting typicality as expected. The final security theorem (Theorem \ref{thm:yay}) is a consequence of each of the above steps in the analysis and leads us surprisingly, to a near-optimal bound on the probability of cloning in the Haar-measure scheme.}

\section{Methods}\label{sec-meth}

\highlight{In this section, we present the technical analysis that establishes our main result.

\subsection{Enforcing unitary invariance} \label{sec:u-inv}

 \begin{definition}
    We say that a cloning attack $\texttt{A}=(B,C,\{B^U_x\}_{U,x},\{C^U_x\}_{U,x},\Phi)$ for the Haar measure encryption $\texttt{Q}_{\mathcal{U}(d),2}$ is \emph{unitarily-invariant} if there exist unitary representations $\pi_B$ and $\pi_C$ of $\mathcal{U}(d)$ on $\mathcal{H}_B$ and $\mathcal{H}_C$, respectively, such that $B^U_x=\pi_B(U)B^I_x\pi_B(U)^\dag$ and $C^U_x=\pi_C(U)C^I_x\pi_C(U)^\dag$.

    Unitarily-invariant cloning attacks are parametrised as $(B,C,\{B^I_x\}_x,\{C^I_x\}_x,\pi_B,\pi_C,\Phi)$; and symmetric unitarily-invariant cloning attacks are parametrised as $(\mathcal{H}_B,\{B^I_x\}_x,\pi_B,\Phi)$
\end{definition}

\begin{lemma}\label{lem:continuous-attack}
	For all $\varepsilon>0$, there exists a symmetric cloning attack $\texttt{A}=(\mathcal{H},\{B^U_x\},\Phi)$ for $\texttt{Q}_{\mathcal{U}(d),2}$ such that the map $U\mapsto B^U_x$ is continuous and $\mathfrak{c}(\texttt{Q}_{\mathcal{U}(d),2},\texttt{A})\geq\mathfrak{c}(\texttt{Q}_{\mathcal{U}(d),2})-\varepsilon.$
\end{lemma}

\begin{proof}
	Let $\varepsilon,\delta>0$. By definition, there exists a (finite-dimensional) cloning attack $\texttt{A}=(B,C,\{B^U_x\},\{C^U_x\},\Phi)$ such that the maps $U\mapsto B^U_x$ and $U\mapsto C^U_x$ are measurable and $\mathfrak{c}(\texttt{Q}_{\mathcal{U}(d),2},\texttt{A})\geq\mathfrak{c}(\texttt{Q}_{\mathcal{U}(d),2})-\varepsilon/2$. We may, without loss of generality, assume that $\text{rank}(B^U_x)=\frac{1}{2}\dim \mathcal{H}_B$ for all $U$ and $x$. By approximating the matrix elements of $B^U_0$ by continuous functions, we know that there exists continuous $U \mapsto f(U)$ such that the $1$-norm $\int \Vert B^U_0-f(U) \Vert dU<\delta$. Next, we may assume that $f(U)$ is hermitian by replacing $f(U)$ with $\frac{1}{2}\left(f(U)+f(U)^\dag\right)$. We diagonalise $f(U)=V_UD_UV_U^\dag$ where the eigenvalues of $D_U$ are in decreasing order. Note in particular that $V_U$ can be chosen to be a continuous function of $U$. Let $D$ be the diagonal projector of rank $\frac{1}{2}\dim\mathcal{H}_B$ and eigenvalues in decreasing order. We know that $\tilde{B}^U_0=V_U D V_U^\dag$ is the closest projector of rank $\frac{1}{2}\dim \mathcal{H}_B$ to $f(U)$ and that it is a continuous function. Since $B^U_0$ is a also a projector of rank $\frac{1}{2}\dim \mathcal{H}_B$, $\Vert \tilde{B}^U_0-f(U)\Vert \leq \Vert B^U_0-f(U)\Vert$. Using the triangle inequality, $\int \Vert B^U_0-\tilde{B}^U_0\Vert dU \leq 2\delta$. We can take $\tilde{B}^U_1=I-\tilde{B}^U_0$ to get $\int \Vert B^U_1-\tilde{B}^U_1 \Vert dU \leq 2 \delta$. In the same way, we can construct $\tilde{C}^U_x$ such that $U\mapsto\tilde{C}^U_x$ is continuous and $\int \Vert C^U_x-\tilde{C}^U_x \Vert dU \leq 2 \delta$. Now, consider the cloning attack $\tilde{\texttt{A}}=(B,C,\{\tilde{B}^U_x\},\{\tilde{C}^U_x\},\Phi)$. The cloning probabilities differ by
	\begin{align*}
		\left\vert\mathfrak{c}(\texttt{Q}_{\mathcal{U}(d),2},\texttt{A})-\mathfrak{c}(\texttt{Q}_{\mathcal{U}(d),2},\tilde{\texttt{A}})\right\vert\leq\int\frac{1}{2}\sum_x \left\Vert B^U_x\otimes C^U_x-\tilde{B}^U_x\otimes\tilde{C}^U_x \right\Vert dU \leq 4\delta.
	\end{align*}
	Taking $\delta=\varepsilon/8$ gives the result.
\end{proof}

\begin{theorem}\label{thm:first-unitarily-invariant}
	For all $\varepsilon>0$, there exists a symmetric unitarily-invariant cloning attack $\texttt{A}$ for $\texttt{Q}_{\mathcal{U}(d),2}$ such that $\mathfrak{c}(\texttt{Q}_{\mathcal{U}(d),2},\texttt{A})\geq\mathfrak{c}(\texttt{Q}_{\mathcal{U}(d),2})-\varepsilon.$
\end{theorem}

\begin{proof}
	Using Lemma \ref{lem:continuous-attack}, let $\texttt{A}=(B,C,\{B^U_x\},\{C^U_x\},\Phi)$ be a cloning attack for $\texttt{Q}_{\mathcal{U}(d),2}$ such that $U\mapsto B^U_x$ and $U\mapsto C^U_x$ are continuous and $\mathfrak{c}(\texttt{Q}_{\mathcal{U}(d),2},\texttt{A})\geq\mathfrak{c}(\texttt{Q}_{\mathcal{U}(d),2})-\varepsilon/2$. Now, construct a unitarily invariant cloning attack $\tilde{A}=(\tilde{B},\tilde{C},\{\tilde{B}^U_x\},\{\tilde{C}^U_x\},\tilde{\Phi})$ as follows: let $\tilde{\mathcal{H}}_B=\mathcal{H}_B\otimes L^2(\mathcal{U}(d))=L^2(\mathcal{U}(d);\mathcal{H}_B)=\int^\oplus_{\mathcal{U}(d)}\mathcal{H}_BdU$, $\tilde{\mathcal{H}}_C=\mathcal{H}_C\otimes L^2(\mathcal{U}(d))$, $\pi=I\otimes\lambda$ where $\lambda$ is the left regular representation $(\lambda(V) f)(U)=f(V^{-1}U)$, $\tilde{B}_b^I=\int^{\oplus}_{\mathcal{U}(d)}B^{U^\dag}_bdU$, $\tilde{C}_b^I=\int^{\oplus}_{\mathcal{U}(d)}C^{U^\dag}_bdU$, and 
	$$\tilde{\Phi}(\rho)=\frac{1}{\mu_{\mathcal{U}(d)}(E)^2}\int\Phi(U\rho U^\dag)\otimes\lambda(U)^\dag \vert \chi_E \rangle \langle \chi_E \vert \lambda(U)\otimes \lambda(U)^\dag \vert \chi_E \rangle \langle \chi_E \vert \lambda(U)dU,$$
	where $E$ is some fixed neighbourhood of $I$. First, note that we have the following equality:
	\begin{align*}
		\int&\pi(U^\dag)^{\otimes 2}\tilde{\Phi}(U\sigma_x U^\dag)\pi(U)^{\otimes 2}dU\\
        &=\frac{1}{\mu_{\mathcal{U}(d)}(E)^2}\iint\Phi(V U\sigma_x U^\dag V^\dag)\otimes \left({\lambda(U^\dag V^\dag) \vert \chi_E \rangle \langle \chi_E \vert \lambda(V U)}\right)^{\otimes 2}dUdV\\
		&=\frac{1}{\mu_{\mathcal{U}(d)}(E)^2}\int\Phi(U\sigma_x U^\dag)|\chi_{U^\dag E}\rangle\!\langle\chi_{U^\dag E}|^{\otimes 2}dU.
	\end{align*}
	Next, fix $\delta>0$. Since $B^U_0$ is continuous on a compact set, it is absolutely continuous, so there exists open $E$ small enough such that, for all $V\in UE$, $\Vert B^U_x-B^V_x \Vert \leq\delta$. Using this, the winning probability of the unitarily-invariant strategy is
	\begin{align*}
		\mathfrak{c}(\texttt{Q}_{\mathcal{U}(d),2},\tilde{\texttt{A}})&=\int_{\mathcal{U}(d)}\frac{1}{2}\sum_{x}\mathrm{Tr}\left[{(\tilde{B}^U_x\otimes\tilde{C}^U_x)\tilde{\Phi}(U\sigma_xU^\dag)}\right]dU\\
		&=\frac{1}{2}\sum_x\mathrm{Tr}\left[{(\tilde{B}^I_x\otimes \tilde{C}^I_x)\int \pi(U^\dag)^{\otimes 2}\tilde{\Phi}(U\sigma_x U^\dag)\pi(U)^{\otimes 2} dU}\right]\\
		&=\frac{1}{2\mu_{\mathcal{U}(d)}(E)^2}\sum_x\int\mathrm{Tr}\left[{(\tilde{B}^I_x\otimes \tilde{C}^I_x)(\Phi(U\sigma_x U^\dag)\otimes|\chi_{U^\dag E}\rangle\!\langle\chi_{U^\dag E}|^{\otimes 2})}\right]dU.
	\end{align*}
	For any pure state $|\psi\rangle\in \mathcal{H}_B\otimes\mathcal{H}_C$, the trace
    \begin{align*}
    		&\mathrm{Tr}\left[(\tilde{B}^I_x\otimes \tilde{C}^I_x)(|\psi\rangle\!\langle\psi|\otimes|\psi_{U^\dag E}\rangle\!\langle\psi_{U^\dag E}|^{\otimes 2})\right]\\
    		&=\iint\langle\psi|B^V_x\otimes C^W_x|\psi\rangle|\psi_{U^\dag E}(V)|^2|\psi_{U^\dag E}(W)|^2dVdW\\
    		&=\iint\limits_{U^\dag E\times U^\dag E} \langle \psi \vert  B^V_x\otimes C^W_x|\psi \rangle dVdW \\
    		&\geq\iint\limits_{U^\dag E\times U^\dag E} \langle \psi \vert B^U_x\otimes C^U_x|\psi \rangle-2\delta dVdW=\mu_{\mathcal{U}(d)}(E)^2 \left(\mathrm{Tr}\left[(B^U_x\otimes C^U_x)|\psi\rangle\!\langle\psi| \right]-2\delta\right).
    	\end{align*}
	This inequality extends to mixed states by taking convex combinations. Hence, taking $\delta=\varepsilon/4$,
	\begin{align*}
		\mathfrak{c}(\texttt{Q}_{\mathcal{U}(d),2},\tilde{\texttt{A}})&\geq\frac{1}{2}\sum_x\int\mathrm{Tr}\left[{(B^U_x\otimes C^U_x)\Phi(U\sigma_x U^\dag)}\right]dU-2\delta\geq\mathfrak{c}(\texttt{Q}_{\mathcal{U}(d),2})-\varepsilon.
	\end{align*}
\end{proof}

\begin{corollary}
	The unitarily-invariant strategy of the theorem can be chosen to be finite-dimensional and pure ($B^I_0$ is a projector and $\Phi$ is an isometry).
\end{corollary}

\begin{proof}
	Let $\texttt{A}=(B,C,\{B^U_x\},\{C^U_x\},\Phi)$ be a unitarily-invariant strategy such that $\mathfrak{c}(\texttt{Q}_{\mathcal{U}(d),2},\tilde{\texttt{A}})\geq\mathfrak{c}(\texttt{Q}_{\mathcal{U}(d),2})-\frac{\varepsilon}{2}$. First, using Stinespring dilation, $\Phi(\rho)=\mathrm{Tr}_R \left(V\rho V^\dag\right)$ for some auxiliary register $R$ and isometry $V:\mathbb{C}^d\rightarrow \mathcal{H}_{BCR}$. Now, due to the Peter-Weyl theorem, $\pi_B$ and $\pi_C\otimes I_R$ decompose as direct sums of finite-dimensional representations. Let $(P_\alpha^B)$ and $(P_\alpha^C)$ be the nets of finite-rank projectors that commute with the representations $\pi_B$ and $\pi_C\otimes I_R$, respectively (ordered by the order on the positive operators). Also, let $\{|e_1\rangle,\ldots,|e_d\rangle\}$ be an orthonormal basis of $\text{im} V$. As we know that $P_\alpha^B\rightarrow I$ and $P^C_\alpha\rightarrow I$ in the strong operator topology, $(P_\alpha^B\otimes P_\alpha^C)|e_i\rangle\rightarrow|e_i\rangle$, giving that there exists $\alpha$ such that $\left\Vert (P^B_\alpha\otimes P^C_\alpha)|e_i\rangle-|e_i\rangle \right\Vert \leq\delta$ for some fixed $\delta>0$ and all $i$. First, for any state $|\psi\rangle=\sum_i\psi_i|e_i\rangle\in \text{im} V$,
	$$\left\Vert (P_\alpha^B\otimes P_\alpha^C)|\psi\rangle-|\psi\rangle \right\Vert \leq\sum_i|\psi_i| \left\Vert (P_\alpha\otimes P_\alpha)|e_i\rangle-|e_i\rangle \right\Vert \leq\sqrt{d}\delta.$$
	Then, for $\rho=\sum_ip_i|\psi_i\rangle\!\langle\psi_i|$ on $\text{im} V$, acting with the projection gives
	\begin{align*}
		\left\Vert (P^B_\alpha\otimes P^C_\alpha)\rho(P^B_\alpha\otimes P^C_\alpha)-\rho \right\Vert_{\mathrm{Tr}}&\leq\sum_ip_i \left\Vert (P_\alpha^B\otimes P_\alpha^C)|\psi_i\rangle\!\langle\psi_i|(P_\alpha^B\otimes P_\alpha^C)-|\psi_i\rangle\!\langle\psi_i| \right\Vert_{\mathrm{Tr}}\\
		&\leq\sum_ip_i\sqrt{2} \left\Vert (P_\alpha^B\otimes P_\alpha^C)\vert \psi_i \rangle -|\psi_i\rangle \right\Vert^2\leq\sqrt{2}d\delta^2.
	\end{align*} 
	Now, we can choose $\delta$ so that $\varepsilon/2=\sqrt{2}d\delta^2$. We get that
	\begin{align*}
		&\int\frac{1}{2}\sum_x\mathrm{Tr}\left[{(\pi_B(U)B^I_x\pi_B(U)^\dag\otimes\pi_C(U)C^I_x\pi_C(U)^\dag\otimes I_R)(P^B_\alpha\otimes P_\alpha^C)VU\sigma_x U^\dag V(P^B_\alpha\otimes P_\alpha^C)}\right]dU\\
        &\geq\mathfrak{c}(\texttt{Q}_{\mathcal{U}(d),2})-\varepsilon.
	\end{align*}
	This is the winning probability of the finite-dimensional unitarily-invariant strategy $(\mathrm{supp}\, P_\alpha^B,\mathrm{supp}\, P_\alpha^C,\{P_\alpha^B B^U_x P_\alpha^B\},\{P_\alpha^C (C^U_x\otimes I_R) P_\alpha^C\},\tilde{\Phi})$, where $\tilde{\Phi}(\rho)=(P_\alpha^B\otimes P_\alpha^C)V\rho V^\dag(P_\alpha^B\otimes P_\alpha^C)$. The final step is to purify the strategy. Using Naimark dilation, the measurements can be made projective. Similarly, $\tilde{\Phi}$ is not a CPTP map, but it is trace non-increasing, so it can be made CPTP by appending a dimension, and then can be purified again using Stinespring dilation.
\end{proof}

\subsection{Bounded-order strategies} \label{sec:bdd-u-inv}

    \begin{lemma}\label{lem:nice-frobenius-bound}
        Let $U,V\in\mathcal{U}(d)$ be unitaries, and let $\rho\in\mathcal{B}(\mathbb{C}^d)$ be a state. Then, $\|U\rho U^\dag-V\rho V^\dag\|_{\mathrm{Tr}}\leq d\|\rho\|\|U-V\|_{f}$.
    \end{lemma}

    \begin{proof}
        For operator $A\in\mathcal{B}(\mathbb{C}^d)$, by Jensen's inequality, the trace norm
        \begin{align*}
            \|A\|_{\mathrm{Tr}}=\frac{1}{2}\mathrm{Tr}(\sqrt{A^\dag A})\leq\frac{1}{2}\sqrt{d\mathrm{Tr}(A^\dag A)}=\frac{d}{2}\|A\|_{f}.
        \end{align*}
        Then, for $A=U\rho U^\dag-V\rho V^\dag$,
        \begin{align*}
            \|A\|_{f}\leq\|U\rho U^\dag-V\rho U^\dag\|_f+\|V\rho U^\dag-V\rho V^\dag\|_f\leq2\|\rho\|\|U-V\|_f,
        \end{align*}
        finishing the proof.
    \end{proof}

    \begin{lemma}[\cite{EY07}]\label{lem:sign-poly}
        For each $a>0$, there exists a polynomial $p$ of degree $m$ such that
        \begin{align*}
            \sup_{x\in[-1,-a]\cup[a,1]}\left|p(x)-\mathrm{sgn}(x)\right|\leq\left(\frac{1-a}{\sqrt{\pi a}}+o_m(1)\right)\frac{1}{\sqrt{m}}\left(\frac{1-a}{1+a}\right)^m,
        \end{align*}
        and $|p(x)|\leq1+\left(\frac{1-a}{\sqrt{\pi a}}+o_m(1)\right)\frac{1}{\sqrt{m}}\left(\frac{1-a}{1+a}\right)^m$ for all $x\in[-1,1]$.
    \end{lemma}

    \begin{lemma}
        The volume of $\mathcal{U}(d)$, seen as a $d^2$-dimensional surface in $\mathbb{M}_d$, is
        \begin{align*}
            \frac{(2\pi)^{\frac{d^2+d}{2}}}{(d-1)!(d-2)!\cdots1!}.
        \end{align*}
    \end{lemma}

    \begin{lemma}[\cite{Bar00}]\label{lem:barnes}
        The \emph{Barnes $G$-function}, which takes values $G(n+1)=(n-1)!(n-2)!\cdots1!$ on the positive integers, has asymptotic expansion
        \begin{align*}
            G(z+1)=C\frac{(2\pi)^{z/2}}{z^{1/12}}\left(\frac{z}{e^{3/2}}\right)^{z^2/2}e^{O(1/z)},
        \end{align*}
        where $C\approx0.848$ is a universal constant.
    \end{lemma}

    \begin{lemma}\label{lem:norm-proximity}
        Let $H\in\mathbb{M}_d$ be a Hermitian matrix such that $-\pi I\leq H\leq\pi I$, and suppose that $\|e^{iH}-I\|_f\leq\varepsilon$. Then,
        \begin{align*}
            \left|\|H\|_f-\|e^{iH}-I\|_f\right|\leq\frac{\pi^2}{16}\sqrt{d}\varepsilon^{2}.
        \end{align*}
    \end{lemma}

    \begin{proof}
        First, we bound $\|H\|_f$ in a relatively naive way. Note that for any $\theta\in(-\pi,\pi]$, $\theta^2\leq\frac{\pi^2}{4}|e^{i\theta}-1|^2$. Therefore, if $-\pi I<H\leq\pi I$, then, taking $\lambda_i$ to be the eigenvalues of $H$,
        \begin{align*}
            \|H\|_f^2=\frac{1}{d}\sum_{i}\lambda_i^2\leq\frac{\pi^2}{4d}\sum_i|e^{i\lambda_i}-1|^2=\frac{\pi^2}{4}\|e^{iH}-I\|_f^2\leq\frac{\pi^2}{4}\varepsilon^2.
        \end{align*}
        Now, we use this to find a tighter bound. By an elementary argument, for $\theta\in(-\pi,\pi]$, $|1+i\theta-e^{i\theta}|^2\leq(\cos\theta-1)^2+(\theta-\sin\theta)^2=\sum_{n=2}^\infty\frac{(-1)^n\theta^{2n}}{n(2n-2)!}\leq\frac{\theta^4}{4}$. As such,
        \begin{align*}
            \left|\|H\|_f-\|e^{iH}-I\|_f\right|^2&\leq\left\|I+iH-e^{iH}\right\|^2_f=\frac{1}{d}\sum_{i}|1+i\lambda_i-e^{i\lambda_i}|^2\leq\frac{1}{4d}\sum_i\lambda_i^4=\frac{1}{4}\|H^2\|_f.
        \end{align*}
        By submultiplicativity of the big Frobenius norm, $\|H^2\|_f\leq\sqrt{d}\|H\|_f^2\leq\frac{\pi^2}{4}\sqrt{d}\varepsilon^2$
    \end{proof}
    
    \begin{lemma}\label{lem:annulus-volume}
        Let $R>r>0$. Let $A=\left\{U\in\mathcal{U}(d)\middle| r<\|U-I\|_f<R\right\}$ be the annulus of inner radius $R$ and outer radius $R$ in the unitary group. Suppose $R+\frac{\pi^2}{16}\sqrt{d}R^2\leq 1$. Then, the Haar measure of $A$ is bounded as
        \begin{align*}
            \mu_{\mathcal{U}(d)}(A)\leq K d^{11/12}e^{-d^2/4}(R-r+\frac{\pi^2}{16}\sqrt{d}(R^2+r^2)),
        \end{align*}
        for $K$ a universal constant.
    \end{lemma}

    \begin{proof}
        First, consider a ball of radius $\varepsilon>0$ in the unitary group $B(\varepsilon)=\left\{U\in\mathcal{U}(d)\middle|\|U-I\|_f\leq\varepsilon\right\}$. There exists a neighbourhood $N$ of $0$ in the Hermitian matrices with eigenvalues between $-\pi$ and $\pi$ such that $B(\varepsilon)=e^{i N}$. By Lemma~\ref{lem:norm-proximity}, $i\ln B(R)$ is contained in the ball in the hermitian matrices of radius $\varepsilon+\frac{\pi^2}{16}\sqrt{d}\varepsilon^2$ and contains the ball of radius $\varepsilon-\frac{\pi^2}{16}\sqrt{d}\varepsilon^2$ (in the little Frobenius norm). Let $b(\delta)$ be the ball of radius $\delta$ in the hermitian matrices. Then, since the volume of an $n$-dimensional sphere of radius $r$ is $\frac{\pi^{n/2}}{\Gamma(n/2+1)}r^n$, the volume of $b(\delta)$ is $\mu_\Lambda(b(\delta))=\frac{\pi^{d^2/2}}{\Gamma(d^2/2+1)}(\sqrt{d}\delta)^{d^2}$, where $\Lambda$ is the Lebesgue measure. As such, the Haar measure of $e^{ib(\delta)}$ is
        \begin{align*}
            \mu_{\mathcal{U}(d)}(e^{ib(\delta)})=\frac{\mu_{\Lambda}(b(\delta))}{\mathrm{vol}(\mathcal{U}(d))}=\frac{\frac{\pi^{d^2/2}}{\Gamma(d^2/2+1)}(\sqrt{d}\delta)^{d^2}}{\frac{(2\pi)^{\frac{d^2+d}{2}}}{(d-1)!(d-2)!\cdots1!}}=\frac{d^{d^2/2}G(d+1)}{2^{\frac{d^2+d}{2}}\pi^{d/2}\Gamma(d^2/2+1)}\delta^{d^2}
        \end{align*}
        Now, we approximate this volume using Lemma~\ref{lem:barnes} and Stirling's approximation: $G(d+1)=C\frac{(2\pi)^{d/2}}{d^{1/12}}\left(\frac{d}{e^{3/2}}\right)^{d^2/2}e^{O(1/d)}$ and $\Gamma(d^2/2+1)=\sqrt{\pi}d\left(\frac{d^2}{2e}\right)^{d^2/2}e^{O(1/d)}$. This gives
        \begin{align*}
            \mu_{\mathcal{U}(d)}(e^{ib(\delta)})&=Ce^{O(1/d)}\frac{d^{d^2/2}\frac{(2\pi)^{d/2}}{d^{1/12}}\left(\frac{d}{e^{3/2}}\right)^{d^2/2}}{2^{\frac{d^2+d}{2}}\pi^{d/2}\sqrt{\pi}d\left(\frac{d^2}{2e}\right)^{d^2/2}}\delta^{d^2}\\
            &=\frac{Ce^{O(1/d)}}{\sqrt{\pi}}d^{-13/12}e^{-d^2/4}\delta^{d^2}.
        \end{align*}
        Take $K=\frac{Ce^{O(1/d)}}{\sqrt{\pi}}$. Then, we have that $\mu_{\mathcal{U}(d)}(B(R))\leq\mu_{\mathcal{U}(d)}(e^{ib(R+\frac{\pi^2}{16}\sqrt{d}R^2)})=Kd^{-13/12}e^{-d^2/4}\left(R+\frac{\pi^2}{16}\sqrt{d}R^2\right)^{d^2}$ and $\mu_{\mathcal{U}(d)}(B(r))\geq\mu_{\mathcal{U}(d)}(e^{ib(r-\frac{\pi^2}{16}\sqrt{d}r^2)})=Kd^{-13/12}e^{-d^2/4}\left(r-\frac{\pi^2}{16}\sqrt{d}r^2\right)^{d^2}$, so
        \begin{align*}
            \mu_{\mathrm{U}(d)}(A)&=\mu_{\mathcal{U}(d)}(B(R))-\mu_{\mathcal{U}(d)}(B(r))\\
            &\leq Kd^{-13/12}e^{-d^2/4}\left[\left(R+\frac{\pi^2}{16}\sqrt{d}R^2\right)^{d^2}-\left(r-\frac{\pi^2}{16}\sqrt{d}r^2\right)^{d^2}\right].
        \end{align*}
        Using the fact that $x^n-y^n=(x^{n-1}+x^{n-2}y+\ldots+xy^{n-2}+y^{n-1})(x-y)$ finishes the proof.
    \end{proof}

    \begin{lemma}\label{lem:indicator-approx}
        Let $E=\left\{U\in\mathcal{U}(d)\middle|\|U-I\|_f< \varepsilon\right\}$, fix $\delta<\varepsilon$ and $m>\frac{1}{\delta}$. There exists a vector $|\psi_E\rangle\in L^2(\mathcal{U}(d))$ such that
        \begin{enumerate}[(i)]
            \item $\psi_E(U)$ is a polynomial in the entries of $U$ and $\overline{U}$, with degree $\leq 2m,$
            \item $\||\chi_E\rangle-|\psi_E\rangle\|\leq Le^{-\delta m/2}$,
        \end{enumerate}
        for some universal constant $L$.
    \end{lemma}

    \begin{proof}
        Set $\psi_E(U)=\frac{1-p(\|U-I\|_f^2-\varepsilon^2)}{2}$, where $p$ is the degree-$m$ polynomial approximation of $\mathrm{sgn}$ on $[-1,-\delta]\cup[\delta,1]$ from Lemma~\ref{lem:sign-poly}. Note that $\psi_E$ has degree $2m$. Evaluated at $U\in\mathcal{U}(d)$,
        \begin{align*}
            (|\chi_E\rangle-|\psi_E\rangle)(U)=\frac{\mathrm{sgn}(\|U-I\|_f^2-\varepsilon^2)-p(\|U-I\|_f^2-\varepsilon^2)}{2}
        \end{align*}
        If $\|U-I\|_f\leq\varepsilon-\delta$ or $\|U-I\|_f\geq\varepsilon+\delta$, then by  Lemma~\ref{lem:sign-poly},
        \begin{align*}
            \left|(|\chi_E\rangle-|\psi_E\rangle)(U)\right|\leq\left(\frac{1-\delta}{2\sqrt{\pi\delta}}+o_m(1)\right)\frac{1}{\sqrt{m}}\left(\frac{1-\delta}{1+\delta}\right)^m.
        \end{align*}
        On the other hand, the Haar measure of the region $A$ of $\mathcal{U}(d)$ where $\varepsilon-\delta<\|U-I\|_f<\varepsilon+\delta$ is upper bounded by, using Lemma~\ref{lem:annulus-volume},
        \begin{align*}
            \mu_{\mathcal{U}(s)}(A)\leq Kd^{11/12}e^{-d^2/4}(2\delta+\frac{\pi^2}{8}\sqrt{d}(\varepsilon^2+\delta^2)).
        \end{align*}
        Note also that for all $U$, $\left|(|\chi_E\rangle-|\psi_E\rangle)(U)\right|\leq2+\left(\frac{1-\delta}{2\sqrt{\pi\delta}}+o_m(1)\right)\frac{1}{\sqrt{m}}\left(\frac{1-\delta}{1+\delta}\right)^m\leq 3$. Therefore, the norm difference
        \begin{align*}
            \||\chi_E\rangle-|\psi_E\rangle\|^2&=\int\left|(|\chi_E\rangle-|\psi_E\rangle)(U)\right|^2dU\\
            &\leq \max_{\substack{\|U-I\|_f\leq\varepsilon-\delta\\\lor\|U-I\|_f\geq\varepsilon+\delta}}\left|(|\chi_E\rangle-|\psi_E\rangle)(U)\right|^2+\mu_{\mathcal{U}(d)}(A)\max_U\left|(|\chi_E\rangle-|\psi_E\rangle)(U)\right|^2\\
            &\leq \left(\frac{1-\delta}{2\sqrt{\pi\delta}}+o_m(1)\right)\frac{1}{\sqrt{m}}\left(\frac{1-\delta}{1+\delta}\right)^m+3Kd^{11/12}e^{-d^2/4}(2\delta+\frac{\pi^2}{8}\sqrt{d}(\varepsilon^2+\delta^2)).
        \end{align*}
        We have that $\left(\frac{1-\delta}{1+\delta}\right)^m\leq e^{-\delta m}$. Since the second term is exponentially small in $d$, it can be neglected, giving the result.
    \end{proof}
    
    \begin{theorem}\label{thm:bounded-unitarily-invariant}
    	There exists a unitarily-invariant cloning attack $\texttt{A}$ for $\texttt{Q}_{\mathcal{U}(d),2}$ such that $\mathfrak{c}(\texttt{Q}_{\mathcal{U}(d),2},\texttt{A})\geq\mathfrak{c}(\texttt{Q}_{\mathcal{U}(d),2})-O(\frac{1}{d^{1/8}})$ where the representations $\pi_B(U)=\pi_C(U)=U^{\otimes m}\otimes\overline{U}^{\otimes m}$ with $m=O(d^{1/4})$.
    \end{theorem}
    
    \begin{proof}
    	Fix $\varepsilon>\delta>0$ and $m>\frac{1}{\delta}$. Using Lemma \ref{thm:first-unitarily-invariant}, let $\texttt{A}=(B,C,\{B^U_x\},\{C^U_x\},\Phi)$ be a unitarily invariant cloning attack for $\texttt{Q}_{\mathcal{U}(d),2}$ such that $\mathfrak{c}(\texttt{Q}_{\mathcal{U}(d),2},\texttt{A})\geq\mathfrak{c}(\texttt{Q}_{\mathcal{U}(d),2})-\varepsilon/2$. We may assume that the representations $\pi_B$ and $\pi_C$ are finite-dimensional and take the form $\pi_B(U)=\pi_C(U)=U^{\otimes N}\otimes\overline{U}^{\otimes M}$ for some $M,N\in\mathbb{N}$; and the channel $\Phi$ is unitarily equivariant. Now, construct a unitarily invariant cloning attack $\tilde{A}=(\tilde{B},\tilde{C},\{\tilde{B}^U_x\},\{\tilde{C}^U_x\},\tilde{\Phi})$ as follows: let $\tilde{\mathcal{H}}_B=\mathcal{H}_B\otimes L^2(\mathcal{U}(d))=L^2(\mathcal{U}(d);\mathcal{H}_B)=\int^\oplus_{\mathcal{U}(d)}\mathcal{H}_BdU$, $\tilde{\mathcal{H}}_C=\mathcal{H}_C\otimes L^2(\mathcal{U}(d))$, $\pi=I\otimes\lambda$ where $\lambda$ is the left regular representation $(\lambda(V) f)(U)=f(V^{-1}U)$, $\tilde{B}_b^I=\int^{\oplus}_{\mathcal{U}(d)}B^{U^\dag}_bdU$, $\tilde{C}_b^I=\int^{\oplus}_{\mathcal{U}(d)}C^{U^\dag}_bdU$, and 
    	$$\tilde{\Phi}(\rho)=\frac{1}{\||\psi_E\rangle\|^4}\int\Phi(U\rho U^\dag)\otimes\lambda(U)^\dag \vert \psi_E \rangle \langle \psi_E \vert \lambda(U)\otimes \lambda(U)^\dag \vert \psi_E \rangle \langle \psi_E \vert \lambda(U)dU,$$
    	where $E=\left\{U\in\mathcal{U}(d)\middle|\|U-I\|_f<\varepsilon\right\}$. First, since $|\psi_E\rangle$ is a degree-$2m$ polynomial in the entries of $U$ and $\overline{U}$, $\lambda(U)$ acts on $|\psi_E\rangle$ as a finite-dimensional representation consisting of irreps with skew Young diagrams of at most $2m$ boxes. Write $|\psi_{UE}\rangle=\lambda(U)|\psi_E\rangle$; since $\lambda(U)|\chi_E\rangle=|\chi_{UE}\rangle$, $|\psi_{UE}\rangle$ is a polynomial approximation of $|\chi_{UE}\rangle$. We have the following equality:
    	\begin{align*}
    		\int\pi(U^\dag)^{\otimes 2}\tilde{\Phi}(U\sigma_x U^\dag)\pi(U)^{\otimes 2}dU&=\frac{1}{\||\psi_E\rangle\|^4}\iint\Phi(V U\sigma_x U^\dag V^\dag)\otimes \left(\lambda(U^\dag V^\dag) \vert \chi_E \rangle \langle \chi_E \vert \lambda(V U)\right)^{\otimes 2}dUdV\\
    		&=\frac{1}{\||\psi_E\rangle\|^4}\int\Phi(U\sigma_x U^\dag)|\psi_{U^\dag E}\rangle\!\langle\psi_{U^\dag E}|^{\otimes 2}dU.
    	\end{align*}
    	Then, the cloning probability of $\tilde{\texttt{A}}$ is
    	\begin{align*}
    		\mathfrak{c}(\texttt{Q}_{\mathcal{U}(d),2},\tilde{\texttt{A}})&=\int_{\mathcal{U}(d)}\frac{1}{2}\sum_{x}\mathrm{Tr}\left[{(\tilde{B}^U_x\otimes\tilde{C}^U_x)\tilde{\Phi}(U\sigma_xU^\dag)}\right]dU\\
    		&=\frac{1}{2}\sum_x\mathrm{Tr}\left[{(\tilde{B}^I_x\otimes \tilde{C}^I_x)\int \pi(U^\dag)^{\otimes 2}\tilde{\Phi}(U\sigma_x U^\dag)\pi(U)^{\otimes 2} dU}\right]\\
    		&=\frac{1}{2\||\psi_E\rangle\|^4}\sum_x\int\mathrm{Tr}\left[(\tilde{B}^I_x\otimes \tilde{C}^I_x)(\Phi(U\sigma_x U^\dag)\otimes|\psi_{U^\dag E}\rangle\!\langle\psi_{U^\dag E}|^{\otimes 2})\right]dU.
    	\end{align*}
    	Using the unitary equivariance of $\Phi$, the difference
        \begin{align*}
            &\left|\mathrm{Tr}((B^V_x\otimes C^W_x)\Phi(U\sigma_x U^\dag))-\mathrm{Tr}((B^U_x\otimes C^U_x)\Phi(U\sigma_x U^\dag))\right|\\
            &\leq\left|\mathrm{Tr}(((B^V_x-B^U_x)\otimes C^W_x)\Phi(U\sigma_x U^\dag))\right|+\left|\mathrm{Tr}((B^U_x\otimes (C^W_x-C^U_x))\Phi(U\sigma_x U^\dag))\right|\\
            &=\left\|\mathrm{Tr}_B((B^U_x\otimes I)\Phi(UV^\dag U\sigma_x U^\dag V U^\dag-U\sigma_x U^\dag))\right\|_{\mathrm{Tr}}+\left\|\mathrm{Tr}_C((I\otimes C^U_x)\Phi(UW^\dag U\sigma_x U^\dag W U^\dag-U\sigma_x U^\dag))\right\|_{\mathrm{Tr}}\\
            &\leq\left\|\Phi(UV^\dag U\sigma_x U^\dag V U^\dag-U\sigma_x U^\dag)\right\|_{\mathrm{Tr}}+\left\|\Phi(UW^\dag U\sigma_x U^\dag W U^\dag-U\sigma_x U^\dag)\right\|_{\mathrm{Tr}}\\
            &\leq\left\|UV^\dag U\sigma_x U^\dag V U^\dag-U\sigma_x U^\dag\right\|_{\mathrm{Tr}}+\left\|UW^\dag U\sigma_x U^\dag W U^\dag-U\sigma_x U^\dag\right\|_{\mathrm{Tr}}\\
            &=\left\|U\sigma_x U^\dag-V\sigma_x V^\dag\right\|_{\mathrm{Tr}}+\left\|U\sigma_x U^\dag-W\sigma_x W^\dag\right\|_{\mathrm{Tr}}\\
            &\leq 2\|U-V\|_f+2\|U-W\|_f,
        \end{align*}
        where the last line is by Lemma~\ref{lem:nice-frobenius-bound}. Therefore, the trace
        \begin{align*}
    		&\mathrm{Tr}\left[(\tilde{B}^I_x\otimes \tilde{C}^I_x)(\Phi(U\sigma_x U^\dag)\otimes|\psi_{U^\dag E}\rangle\!\langle\psi_{U^\dag E}|^{\otimes 2})\right]\\
    		&=\iint\mathrm{Tr}((B^V_x\otimes C^W_x)\Phi(U\sigma_x U^\dag))(|\psi_{U^\dag E}(V)|^2|\psi_{U^\dag E}(W)|^2dVdW\\
    		&\geq\iint\mathrm{Tr}((B^U_x\otimes C^U_x)\Phi(U\sigma_x U^\dag))-2\|U-V\|_f-2\|U-W\|_f)|\psi_{U^\dag E}(V)|^2|\psi_{U^\dag E}(W)|^2dVdW\\
            &=\|\psi_E\|^4\mathrm{Tr}((B^U_x\otimes C^U_x)\Phi(U\sigma_x U^\dag))-4\|\psi_E\|^2\int\|U-V\|_f|\psi_{U^\dag E}(V)|^2dV
    	\end{align*}
        This means that the winning probability
    	\begin{align*}
    		\mathfrak{w}(\texttt{Q}_{\mathcal{U}(d),2},\tilde{\texttt{A}})&\geq\mathfrak{w}(\texttt{Q}_{\mathcal{U}(d),2},\texttt{A})-\frac{4}{\|\psi_E\|^2}\int\|I-V\|_f|\psi_{E}(V)|^2dV,
    	\end{align*}
        and by Lemma~\ref{lem:indicator-approx}, there exists a universal constant $K$ such that $$\frac{4}{\|\psi_E\|^2}\int\|I-V\|_f|\psi_{E}(V)|^2dV\leq K(e^{-\delta m/2}+\varepsilon).$$ Now, take $m=d^{1/4}$, $\delta=\frac{\lg d}{4d^{1/4}}$ and $\varepsilon=d^{1/8}$, which gives the result.
    \end{proof}

\subsection{Approximating unitarily-invariant attacks} \label{sec:half-deFin}

\begin{lemma}\label{lem:singlet-span}
    The space of vectors invariant under the representation $U\mapsto U^{\otimes n}\otimes \overline{U}^{\otimes n}$ of $\mathcal{U}(d)$ is spanned by states of the form $\bigotimes_{i=1}^n\left|\phi^+\right\rangle_{i,n+\sigma(i)}$, for all permutations $\sigma$ of $n$.
\end{lemma}

\begin{proof}
    By Schur-Weyl duality, the space of operators intertwining the representation $U\mapsto U^{\otimes n}$ is spanned by $V_d(\sigma)$ for $\sigma\in S_n$, where $V_d$ is the representation of $S_n$ given by $V_d(\sigma)\left|x_1\cdots x_n\right\rangle=|x_{\sigma^{-1}(1)}\cdots x_{\sigma^{-1}(n)}\rangle$. Every vector in $|{v}\rangle\in (\mathbb{C}^{d})^{\otimes 2n}$ can be expressed as $|{v}\rangle=(\lambda\otimes I)|{\phi^+}\rangle$ for some operator $\lambda$ on $(\mathbb{C}^d)^{\otimes n}$. Therefore, if $|{v}\rangle$ is invariant under $U^{\otimes n}\otimes \overline{U}^{\otimes n}$, then $\lambda$ intertwines $U^{\otimes n}$. Therefore, the space of invariant vectors is spanned by
    \begin{align*}
        (V_d(\sigma)\otimes I)|\phi^+\rangle=\frac{1}{\sqrt{d^n}}\sum_{x_1,\ldots,x_n}^{d}|x_{\sigma^{-1}(1)}\cdots x_{\sigma^{-1}(n)}\rangle\otimes|x_1\cdots x_n\rangle=\bigotimes_{i=1}^n|\phi^+\rangle_{i,n+\sigma(i)}.
    \end{align*}
\end{proof}

\begin{lemma}\label{lem:cross-terms}
    Let $\mathcal{H}_A$, $\mathcal{H}_B$, and $\mathcal{H}_C$ be $d=2^n$-dimensional Hilbert spaces, and let $\mathcal{H}_{B'}$ and $\mathcal{H}_{C'}$ be arbitrary finite-dimensional Hilbert spaces. Let $\vert \psi \rangle \in\mathcal{H}_{B'CC'}$ and $\vert \psi' \rangle \in\mathcal{H}_{BB'C'}$ be arbitrary finite-dimensional states. Let $P\in\mathcal{B}(\mathcal{H}_{BB'})$ and $Q\in\mathcal{B}(\mathcal{H_{CC'}})$ be positive semidefinite operators with operator norm $\leq 1$. For $b\in\{0,1\}$, write $A_b=\vert b \rangle \langle b \vert \otimes I$. Then,
    \begin{align*}
        \left|(\langle\phi^+|_{AB}\otimes\langle\psi|_{B'CC'})(A_b\otimes P\otimes Q)(|\phi^+\rangle_{AC}\otimes|\psi'\rangle_{BB'C'})\right|\leq\frac{1}{d}.
    \end{align*}
\end{lemma}

\begin{proof}
    Write $p=(\langle\phi^+|_{AB}\otimes\langle\psi|_{B'CC'})(A_b\otimes P\otimes Q)(|\phi^+\rangle_{AC}\otimes|\psi'\rangle_{BB'C'})$. Then, we expand
    \begin{align*}
        p&=\frac{1}{d}\sum_{x,y=1}^d\langle x|(|b\rangle\!\langle b|\otimes I)|y\rangle(\langle x|_{B}\otimes\langle\psi|_{B'CC'})(P\otimes Q)(|y\rangle_{C}\otimes|\psi'\rangle_{BB'C'})\\
        &=\frac{1}{d}\sum_{x'=1}^{d/2}(\langle bx'|_{B}\otimes\langle\psi|_{B'CC'})(P\otimes Q)(|bx'\rangle_{C}\otimes|\psi'\rangle_{BB'C'}).
    \end{align*}
    We expand $P$ and $Q$ in the canonical bases of $\mathcal{H}_B$ and $\mathcal{H}_C$, respectively, as $P=\sum_{x,y=1}^d|x\rangle\!\langle y|\otimes P_{x,y}$ and $Q=\sum_{x,y=1}^d|x\rangle\!\langle y|\otimes Q_{x,y}$. Then,
    \begin{align*}
        p&=\frac{1}{d}\sum_{x'=1}^{d/2}\sum_{x,y,z,w=1}^d(\langle bx'|_{B}\otimes\langle\psi|_{B'CC'})(|x\rangle\!\langle y|_B\otimes (P_{x,y})_{B'}\otimes |z\rangle\!\langle w|_{C}\otimes (Q_{z,w})_{C'})(|bx'\rangle_{C}\otimes|\psi'\rangle_{BB'C'})\\
        &=\frac{1}{d}\sum_{x'=1}^{d/2}\sum_{y,z=1}^d\langle\psi|_{B'CC'}(|z\rangle_{C}\otimes I_{B'C'})( (P_{bx',y})_{B'}\otimes (Q_{z,bx'})_{C'})(\langle y|_B\otimes I_{B'C'})|\psi'\rangle_{BB'C'}\\
        &=\frac{1}{d}\sum_{x'=1}^{d/2}\langle\psi|_{B'BC'}(P_{BB'}\otimes I_{C'})(|bx'\rangle\!\langle bx'|_{B}\otimes I_{B'C'})(Q_{BC'}\otimes I_{B'})|\psi'\rangle_{BB'C'}\\
        &=\frac{1}{d}\langle\psi|_{B'BC'}(P_{BB'}\otimes I_{C'})((A_0)_B\otimes I_{B'C'})(Q_{BC'}\otimes I_{B'})|\psi'\rangle_{BB'C'}.
    \end{align*}
    Since
    \begin{align*}
        &\left|\langle\psi|_{B'BC'}(P_{BB'}\otimes I_{C'})((A_0)_B\otimes I_{B'C'})(Q_{BC'}\otimes I_{B'})|\psi'\rangle_{BB'C'}\right|\\
        &\qquad\qquad\leq\left\|(P_{BB'}\otimes I_{C'})((A_0)_B\otimes I_{B'C'})(Q_{BC'}\otimes I_{B'})\right\|\leq 1,
    \end{align*}
    we get as wanted that $|p|\leq\frac{1}{d}$.
\end{proof}

\begin{lemma}\label{lem:middle-terms}
    Let $\mathcal{H}_A$ and $\mathcal{H}_{B_1},\ldots,\mathcal{H}_{B_n}$ be $d=2^n$-dimensional Hilbert spaces, and let $\mathcal{H}_{B'}$ and $\mathcal{H}_C$ be arbitrary finite-dimensional Hilbert spaces. For $i=1,\ldots,n$, let $|\psi_i\rangle\in\mathcal{H}_{\hat{B}_iB'C}$ be pure states, where $\hat{B}_i$ is the register $B_1\cdots B_{i-1}B_{i+1}\cdots B_n$. Let $|\psi\rangle=\sum_{i=1}^nc_i|\phi^+\rangle_{AB_i}\otimes|\psi_i\rangle_{\hat{B}_iB'C}$ be a pure state. Let $A_b$ be as in the previous lemma, and let $\{B_0,B_1\}$ and $\{C_0,C_1\}$ be POVMs on $\mathcal{H}_{B_1\cdots B_n B'}$ and $\mathcal{H}_{C}$, respectively. Then,
    \begin{align*}
        \sum_{b=0}^1\langle\psi|A_b\otimes B_b\otimes C_b|\psi\rangle\leq\frac{1}{2}+\frac{2n^2}{d}
    \end{align*}
\end{lemma}

\begin{proof}
    First, we use the simple upper bound
    \begin{align*}
        \sum_{b=0}^1&\langle\psi|A_b\otimes B_b\otimes C_b|\psi\rangle\leq\sum_{b=0}^1\langle\psi|A_b\otimes I\otimes C_b|\psi\rangle\\
        &=\sum_{i,j=1}^n\sum_{b=0}^1\bar{c}_ic_j(\langle\phi^+|_{AB_i}\otimes\langle\psi_i|_{\hat{B}_iB'C})(A_b\otimes I\otimes C_b)(|\phi^+\rangle_{AB_j}\otimes|\psi_i\rangle_{\hat{B}_jB'C})\\
        &\leq\sum_{i,j=1}^n|c_i||c_j|\sum_{b=0}^1\left|(\langle\phi^+|_{AB_i}\otimes\langle\psi_i|_{\hat{B}_iB'C})(A_b\otimes I\otimes C_b)(|\phi^+\rangle_{AB_j}\otimes|\psi_j\rangle_{\hat{B}_jB'C})\right|.
    \end{align*}
    Let $p_{ij}=\sum_{b=0}^1\left|(\langle\phi^+|_{AB_i}\otimes\langle\psi_i|_{\hat{B}_iB'C})(A_b\otimes I\otimes C_b)(|\phi^+\rangle_{AB_j}\otimes|\psi_j\rangle_{\hat{B}_jB'C})\right|$. Consider first the case $i=j$. Then, writing $\rho_i=\mathrm{Tr}_{\hat{B_i}B'}(|\psi_i\rangle\!\langle\psi_i|)$,
    \begin{align*}
        p_{ii}&=\sum_{b=0}^1(\langle\phi^+|_{AB_i}\otimes\langle\psi_i|_{\hat{B}_iB'C})(A_b\otimes I\otimes C_b)(|\phi^+\rangle_{AB_i}\otimes|\psi_i\rangle_{\hat{B}_iB'C})\\
        &=\sum_{b=0}^1\frac{1}{d}\mathrm{Tr}(A_b)\mathrm{Tr}(C_b\rho_i)=\frac{1}{2}\sum_{b=0}^1\mathrm{Tr}(C_b\rho_i)=\frac{1}{2}.
    \end{align*}
    Now, consider the case $i\neq j$. Then,
    \begin{align*}
        p_{ij}&=\frac{1}{d}\sum_{b=0}^1\left|\sum_{x,y=1}^d\langle x|A_b|y\rangle(\langle x|_{B_i}\otimes\langle\psi_i|_{\hat{B}_iB'C})(I\otimes C_b)(|y\rangle_{B_j}\otimes|\psi_j\rangle_{\hat{B}_jB'C})\right|\\
        &=\frac{1}{d}\sum_{b=0}^1\left|\sum_{x'=1}^{d/2}(\langle bx'|_{B_i}\otimes\langle\psi_i|_{\hat{B}_iB'C})(I\otimes C_b)(|bx'\rangle_{B_j}\otimes|\psi_j\rangle_{\hat{B}_jB'C})\right|\\
        &=\frac{1}{d}\sum_{b=0}^1\left|\sum_{x'=1}^{d/2}\langle\psi_i|_{\hat{B}_iB'C}(|bx'\rangle\!\langle bx'|_{B_j}\otimes I\otimes (C_b)_C)|\psi_j\rangle_{\hat{B}_iB'C}\right|\\
        &=\frac{1}{d}\sum_{b=0}^1\left|\langle\psi_i|_{\hat{B}_iB'C}((A_b)_{B_j}\otimes I\otimes (C_b)_C)|\psi_j\rangle_{\hat{B}_iB'C}\right|\\
        &\leq\frac{1}{d}\sum_{b=0}^1\left\|(A_b)_{B_j}\otimes I\otimes (C_b)_C\right\|\leq\frac{2}{d}.
    \end{align*}
    Therefore,
    \begin{align*}
        \sum_{b=0}^1\langle\psi|A_b\otimes B_b\otimes C_b|\psi\rangle&\leq\sum_{b=0}^1\langle\psi|A_b\otimes I\otimes C_b|\psi\rangle\leq\sum_{i,j=1}^n|c_i||c_j|p_{ij}\\
        &\leq\sum_{i=1}^n|c_i|^2\frac{1}{2}+\sum_{i\neq j}|c_i||c_j|\frac{2}{d}\\
        &\leq\frac{1}{2}+\frac{2n^2}{d}.
    \end{align*}
\end{proof}

\begin{theorem}\label{thm:half-de-finetti}
    Let $\texttt{S}$ be a unitarily invariant strategy for the $d$-dimensional Haar-measure MoE game where Bob's representation is $\pi_B(U)=U^{\otimes m}\otimes\overline{U}^{\otimes n}$. Then, the winning probability is upper bounded by
    \begin{align*}
        \frac{1}{2}+\frac{2(m+n+1)^2}{d}.
    \end{align*}
\end{theorem}

\begin{proof}
    Due to the unitary invariance, we may assume that the shared state $\rho_{ABC}$ intertwines the representation $\overline{U}\otimes\pi_B(U)\otimes\pi_C(U)$. Next, by convexity, we may assume that $\rho_{ABC}$ is pure. The marginal $\rho_{AB}$ intertwines $\overline{U}\otimes (U^{\otimes m}\otimes\overline{U}^{\otimes n})$, so it admits a purification that is invariant under $\overline{U}\otimes(U^{\otimes m}\otimes\overline{U}^{\otimes n})\otimes(U^{\otimes n+1}\otimes\overline{U}^{\otimes m})$. By Uhlmann's theorem, we may assume that the shared state is this purification. Denote it $|\psi\rangle_{ABC}$.

    Write $B_i$ and $C_i$ for the $i$-th registers of $B$ and $C$, respectively. Using Lemma~\ref{lem:singlet-span}, there exist pure states $\vert \psi_i \rangle \in \mathcal{H}_{\hat{B}_iC}$ and $\vert \psi_i' \rangle \in \mathcal{H}_{B\hat{C}_i}$ and coefficients $c_i,c_i' \in \mathbb{C}$ such that
    \begin{align*}
        \vert \psi \rangle =\sum_{i=1}^{m}c_i|\phi^+\rangle_{AB_i}\otimes|\psi_i\rangle_{\hat{B}_iC}+\sum_{i=1}^{n+1}c_i'|\phi^+\rangle_{AC_i}\otimes|\psi_i'\rangle_{B\hat{C}_i}.
    \end{align*}
    Write $\vert \psi_B \rangle =\sum_{i=1}^{m}c_i|\phi^+\rangle_{AB_i}\otimes|\psi_i\rangle_{\hat{B}_iC}$ and $|\psi_C\rangle=\sum_{i=1}^{n+1}c_i'|\phi^+\rangle_{AC_i}\otimes|\psi_i'\rangle_{B\hat{C}_i}$. Then, the winning probability can be expressed
    \begin{align*}
        \sum_{b=0}^1\langle\psi|A_b\otimes B_b\otimes C_b|\psi\rangle&=\sum_{b=0}^1\langle\psi_B|A_b\otimes B_b\otimes C_b|\psi_B\rangle+\sum_{b=0}^1\langle\psi_B|A_b\otimes B_b\otimes C_b|\psi_C\rangle\\
        &\qquad+\sum_{b=0}^1\langle\psi_C|A_b\otimes B_b\otimes C_b|\psi_B\rangle+\sum_{b=0}^1\langle\psi_C|A_b\otimes B_b\otimes C_b|\psi_C\rangle.
    \end{align*}
    Using Lemma~\ref{lem:cross-terms},
    \begin{align*}
        \sum_{b=0}^1&\langle\psi_B|A_b\otimes B_b\otimes C_b|\psi_C\rangle\\
        &=\sum_{b=0}^1\sum_{i=1}^m\sum_{j=1}^{n+1}\bar{c}_ic_j(\langle\phi^+|_{AB_i}\otimes\langle\psi_i|_{\hat{B}_iC})(A_b\otimes B_b\otimes C_b)(|\phi^+\rangle_{AC_j}\otimes|\psi_i'\rangle_{B\hat{C_j}})\\
        &\leq\sum_{b=0}^1\sum_{i=1}^m\sum_{j=1}^{n+1}|c_i||c_j|\frac{1}{d}\\
        &\leq\frac{2m(n+1)}{d}.
    \end{align*}
    Identically, $\sum_{b=0}^1\langle\psi_C|A_b\otimes B_b\otimes C_b|\psi_B\rangle\leq\frac{2m(n+1)}{d}$. Also, using Lemma~\ref{lem:middle-terms}, we have $\sum_{b=0}^1\langle\psi_B|A_b\otimes B_b\otimes C_b|\psi_B\rangle\leq\frac{1}{2}\langle\psi_B|\psi_B\rangle+\frac{2m^2}{d}$ and $\sum_{b=0}^1\langle\psi_C|A_b\otimes B_b\otimes C_b|\psi_C\rangle\leq\frac{1}{2}\langle\psi_C|\psi_C\rangle+\frac{2(n+1)^2}{d}$. Putting these together,
    \begin{align*}
        \sum_{b=0}^1\langle\psi|A_b\otimes B_b\otimes C_b|\psi\rangle&\leq\frac{1}{2}\langle\psi_B|\psi_B\rangle+\frac{2m^2}{d}+\frac{1}{2}\langle\psi_C|\psi_C\rangle+\frac{2(n+1)^2}{d}+2\frac{2m(n+1)}{d}\\
        &=\frac{1}{2}+\frac{2(m+n+1)^2}{d}.
    \end{align*}
\end{proof}

\subsection{Security proof}

\begin{theorem}[Security] \label{thm:yay}
    The success probability of any cloning attack on $\texttt{Q}_{\mathcal{U}(d),2}$ is\\$\frac{1}{2}+O(\frac{1}{d^{1/8}})$.
\end{theorem}

\begin{proof}
    Using Theorem~\ref{thm:bounded-unitarily-invariant}, there exists a cloning attack $\texttt{A}$ for $\texttt{Q}_{\mathcal{U}(d),2}$ such that the representations are $U^{\otimes m}\otimes \overline{U}^{\otimes m}$ for $m=O(d^{1/4})$, and $\mathfrak{c}(\texttt{Q}_{\mathcal{U}(d),2},\texttt{A})\geq\mathfrak{c}(\texttt{Q}_{\mathcal{U}(d),2})-O(\frac{1}{d^{1/8}})$. Next, using Theorem~\ref{thm:half-de-finetti}, $\mathfrak{c}(\texttt{Q}_{\mathcal{U}(d),2},\texttt{A})\leq\frac{1}{2}+\frac{2(2m+1)^2}{d}=\frac{1}{2}+O(\frac{1}{\sqrt{d}})$. Putting these together, we get $\mathfrak{c}(\texttt{Q}_{\mathcal{U}(d),2})\leq\frac{1}{2}+O(\frac{1}{d^{1/8}})$.
\end{proof}

\begin{corollary}\label{cor:yayay}
    The success probability of any cloning attack on $\texttt{Q}_{\mathcal{U}(2^n),2}$ approaches $\frac{1}{2}$ exponentially in $n$.
\end{corollary}}

\section{Discussion}\label{sec-disc}

\highlight{The uncloneable bit is a fundamental primitive in uncloneable cryptography. Many uncloneable cryptographic functionalities are rendered secure if the uncloneable bit is secure. Our result establishes this conclusion unconditionally, and in full generality. Notably, it also achieves a near-optimal bound on the Haar measure encryption scheme which itself is minimal \cite{BCR25}. Very few cryptographic primitives enjoy statistical (or information-theoretic) security and our proof of security for the uncloneable bit adds to this repertoire. At this point, one may ask further questions about the efficiency of implementing the scheme proven secure. It is interesting that the Haar measure encryption scheme that admits \textit{weak uncloneable security} in \cite{BC26} also begets an efficient implementation via unitary $2$-designs as a consequence of the decoupling-based approach which requires only second moment identities of the Haar integral. Our security analysis for the same scheme establishes unconditional \textit{strong uncloneable security} and achieves full strength, however, the reliance of our techniques on the full unitary group and the truly Haar-random unitaries leads to the absence of an efficient implementation. We suspect it may be possible to achieve the same level of security with the Haar encryption scheme by using pseudorandom unitaries, but now offering computational security due to pseudorandom unitaries. Nevertheless, for practical considerations this may be a useful avenue, and so far the security purported to hold for all uncloneable primitives is also computational. We leave this for future work. 

From the point of view of entanglement, one may ask what the existence of the uncloneable bit implies. First, what we have called the generalised devious states span the unitarily invariant space of all attacks, and are mathematically well-defined, however, they lack an operational interpretation in the context of the unitary invariance symmetry which is pervasive in quantum information science. In analogy with a de Finetti state that approximates arbitrary permutation invariant states, how could one interpret these states symmetric under unitary invariance? Second, it is perhaps obvious by now and also due to the nature of monogamy-of-entanglement games, that maximally entangled states perform poorly for the adversaries \cite{BC26}. Specifically, our analysis confirms that the entanglement present in the tripartite input state is not maximal but nevertheless there is entanglement. This begs the question of what type of entanglement is present. In quantum cryptography, there is evidence of non-maximally entangled states such as bipartite bound entangled states (entangled states with zero distillable entanglement) allowing distillation of secret key \cite{HHHO05}. Our work hints at the relevance of non-maximal entanglement in cryptography or more generally in quantum communication. Any discovery on this front would have significant benefits. For instance, it could enable an explicit construction of a multipartite non-maximally entangled state that is useful for quantum cryptography, and even demonstrate advantage beyond key distribution. Moreover, it is well-known that the Peres-Horodecki criterion~\cite{Per96,HHH01} exclusively applies to bipartite bound entangled states, so computing the distillable entanglement of the tripartite entangled state that is not maximal in uncloneable encryption schemes may reveal answers to elusive questions surrounding the nature of entanglement.}

\backmatter

\bmhead{Acknowledgements}

We thank Takashi Yamakawa for helpful feedback on an earlier version of this paper. A.Bh. thanks Matthias Christandl and Felix Leditzky for helpful discussions. A.Bh. and A.Br. acknowledge the support of the Natural Sciences and Engineering Research Council of Canada (NSERC)(ALLRP-578455-2022, RGPIN-2022-05167), the Air Force Office of Scientific Research under award number FA9550-20-1-0375 and of the Canada Research Chairs Program (CRC-2023-00173). A.Bh. thanks Institut Mittag-Leffler for hospitality during part of this work. E.C. is supported by a \mbox{CGS D} scholarship from NSERC. Research at Perimeter Institute is supported in part by the Government of Canada through the Department of Innovation, Science, and Economic Development Canada and by the Province of Ontario through the Ministry of Colleges and Universities.

\bmhead{Author contributions}

All authors contributed to the formulation, idea, approach, analysis, interpretation and presentation of the results. A.Bh. and E.C. performed the analysis and wrote the manuscript.

\bmhead{Competing interests}

The authors declare no competing interests.

\bmhead{Data availability}

No datasets were generated during this study.

\bmhead{Code availability}

No code was generated or used in the analysis and performance of this work.

\bmhead{Inclusion and Ethics statement}

All authors are equal stakeholders in this work and have contributed to the formulation, idea, analysis, interpretation, and presentation of the results. This research does not pose any risk to any individual or groups. This research was conducted via local collaboration at the host institution of the authors as well as remotely. Two of the authors represent a minority in STEM.

\bmhead{Correspondence and requests for materials} 

All requests should be addressed to Archishna Bhattacharyya, Anne Broadbent, or Eric Culf.

%\bibliography{sn-bibliography}

\makeatletter\@ifundefined{url}{\newcommand{\url}[1]{\texttt{#1}}}{}\@ifundefined{href}{\newcommand{\href}[2]{\texttt{#2}}}{}\@ifundefined{mathbb}{\newcommand{\mathbb}[1]{#1}}{}\makeatother

\pagebreak

\section{Supplementary information: Technical preliminaries}

In this section we present technical preliminaries required to follow the main results, namely the security definition of uncloneable encryption and representation theoretic methods.

\subsection{Notation}

For $n\in \mathbb{N}$, write $[n]=\{1,2,\ldots,n\}$. We write $\log$ for the base-$2$ logarithm. For functions $f,g:\mathbb{N}\rightarrow\mathbb{R}_{\geq 0}$, we say $f(\lambda)=O(g(\lambda))$ if $\lim \limits_{\lambda\rightarrow\infty}\frac{f(\lambda)}{g(\lambda)}<\infty$; and $f(\lambda)=\widetilde{O}(g(\lambda))$ if $f(\lambda)=O(g(\lambda)\log (\lambda)^c)$ for some $c\in\mathbb{R}$. We say a function $f$ is negligible if $\lim \limits_{\lambda\rightarrow\infty}\lambda^n f(\lambda)=0$ for all $n\in\mathbb{N}$.

We denote registers by uppercase Latin letters $A,B,C,\ldots$; and we denote Hilbert spaces by uppercase script letters $\mathcal{H},\mathcal{K},\mathcal{L},\ldots$. We always assume registers are finite sets and Hilbert spaces are finite-dimensional. We denote an independent copy of a register $A$ by $A'$. Given a register $A$, the Hilbert space spanned by $A$ is $\mathcal{H}_A= \text{span} \{\vert a \rangle ~\vert~ a\in A\}\cong\mathbb{C}^{|A|}$. We indicate that an operator or vector is on register $A$ with a subscript $A$, omitting when clear from context. Given two registers $A$ and $B$, we write $AB$ for their Cartesian product, and treat the isomorphism $\mathcal{H}_{AB}\cong\mathcal{H}_A\otimes\mathcal{H}_B$ implicitly. Given finite-dimensional Hilbert spaces $\mathcal{H}$ and $\mathcal{K}$, we write $B(\mathcal{H},\mathcal{K})$ for the set of all linear operators $\mathcal{H}\rightarrow\mathcal{K}$, $B(\mathcal{H})=B(\mathcal{H},\mathcal{H})$, $\mathcal{U}(\mathcal{H})\subseteq B(\mathcal{H})$ for the subset of unitary operators, and $D(\mathcal{H})\subseteq B(\mathcal{H})$ for the subset of density operators. Write $\mathcal{U}(d)=\mathcal{U}(\mathbb{C}^{d})$. We write $\text{Tr}$ for the trace on $B(\mathcal{H})$. On $B(\mathcal{H}_{AB})$, we write the partial trace $\text{Tr}_{B}=\text{id}\otimes\text{Tr}$. For $\rho_{AB}\in B(\mathcal{H}_{AB})$, write $\rho_A=\text{Tr}_B(\rho_{AB})$. We denote the $1$-norm by $\Vert{\cdot}\Vert_1$ and the trace norm by $\Vert{\cdot}\Vert_{\text{Tr}}=\frac{1}{2}\Vert{\cdot}\Vert_1$.

We denote the canonical maximally-entangled state $\vert {\phi^+} \rangle_{AA'}=\frac{1}{\sqrt{|A|}}\sum_{a\in A}\vert{a}\rangle\otimes\vert{a}\rangle\in\mathcal{H}_{AA'}$. We write the maximally-mixed state on a register $A$ as $\omega_A=\frac{1}{|A|}\sum_{a\in A}\vert{a}\rangle\!\langle{a}\vert\in D(\mathcal{H}_A)$. 

A positive-operator-valued measurement (POVM) is a finite set of positive operators $\{P_i\}_{i\in I}$ such that $\sum_iP_i={1}$, and a projection-valued measurement (PVM) is a POVM where all the elements are projectors. A quantum channel is a completely positive trace-preserving (CPTP) map $\Phi:B(\mathcal{H})\rightarrow B(\mathcal{K})$. We denote the Choi-Jamio\l{}kowski isomorphism $J:B(B(\mathcal{H}_A),B(\mathcal{H}_B))\rightarrow B(\mathcal{H}_{AB})$, $J(\Phi)=(\text{id}\otimes\Phi)(\vert{\phi^+}\rangle\!\langle{\phi^+}\vert_{AA'})$. Note that if is $\Phi$ is a quantum channel, $J(\Phi)\in D(\mathcal{H}_{AB})$, called the Choi state. The Kraus representation of a quantum channel $\Phi: B(\mathcal{H}) \to B(\mathcal{H})$ is $\sum \limits_{i = 0}^{d - 1} \Phi(\rho) = A_i \rho A_i^{\dagger}$ where $A_i \in B(\mathcal{H})$ are the Kraus operators such that $\sum \limits_{i = 0 }^{d - 1} A_i^{\dagger} A_i = I,$ and $d = \text{dim}(\mathcal{H})$. For $\Phi: B(\mathcal{H}_A) \to B(\mathcal{H}_B)$, by Stinespring's dilation theorem, there exists an isometry $V: A \to BE$ known as the Stinespring isometry such that $\Phi (\rho) = \text{Tr}_E (V \rho V^{\dagger})$. The complementary channel $\Phi^c: B(\mathcal{H}_A) \to B(\mathcal{H}_E)$ is given by $\Phi^c (\rho) = \text{Tr}_B (V \rho V^{\dagger}).$  A quantum channel $\mathcal{N}$ is said to be degradable if $\exists$ a channel $\mathcal{D}$, known as the degrading map such that $\mathcal{D} \circ \mathcal{N} = \mathcal{N}^c$, where $\mathcal{N}^c$ is the complementary channel of $\mathcal{N}.$ Similarly, a channel is said to be anti-degradable if $\exists$ a channel $\mathcal{A}$, known as the anti-degrading map such that $\mathcal{A} \circ \mathcal{N}^c = \mathcal{N}.$

For a complex-valued random variable $X$, we write its expectation as $\mathbb{E} X=\mathbb{E}_XX$, and its variance as $\varsigma_X^2=\mathbb{E} |X|^2-\left\vert{\mathbb{E} X}\right\vert^2$. We make use of the Haar measure on the unitary group, which is the unique invariant Borel probability measure on $\mathcal{U}(\mathcal{H})$, for $\mathcal{H}$ a finite-dimensional Hilbert space. We denote it $\mu_{\mathcal{U}(\mathcal{H})}$. Given a function $f$ with domain $\mathcal{U}(\mathcal{H})$, we interchangeably write $\mathbb{E}_U f(U)=\int f(U)dU$ for the expectation with respect to the Haar measure.

\subsection{Representation theory}\label{sec:representations}

Let $G$ be a finite or a compact topological group. A \textit{unitary representation} of $G$ on a finite-dimensional Hilbert space $\mathcal{H}$ is a group homomorphism $\pi:G\rightarrow \mathcal{U}(\mathcal{H})$. If $G$ is a topological group, we will also require that $\pi$ be continuous. An \textit{intertwiner} from a representation $\pi:G\rightarrow \mathcal{U}(\mathcal{H})$ to a representation $\chi:G\rightarrow \mathcal{U}(\mathcal{K})$ is an operator $T\in B(\mathcal{H},\mathcal{K})$ such that $\chi(g)T=T\pi(g)$ for all $g\in G$. A natural way to construct intertwiners is by means of the Haar measure on $G$, $\mu_G$. In fact, if $T\in B(\mathcal{H},\mathcal{K})$, $$\int\chi(g)T\pi(g)^\dag d\mu_G(g)$$ is always an intertwiner. We say two representations are \textit{equivalent} if there is an invertible intertwiner between them and write $\pi\simeq\chi$. An \textit{irreducible representation} (irreps) is a representation whose action on $\mathcal{H}$ leaves no subspace but $\mathcal{H}$ and $0$ invariant. By Schur's lemma, the intertwiners between inequivalent irreducible representations are $0$ and the intertwiners from an irreducible representation to itself are multiples of identity. For finite groups (Maschke's theorem) or compact topological groups (Peter-Weyl theorem), every representation decomposes as a direct sum of irreducibles, \textit{i.e.} given a representation $\pi:G\rightarrow \mathcal{U}(\mathcal{H})$ there exists an equivalence $\mathcal{H}\rightarrow\bigoplus_i\mathcal{H}_i\otimes\mathcal{K}_i$ such that $\pi\simeq\bigoplus_i\pi_i\otimes I$, where the $\pi_i:G\rightarrow \mathcal{U}(\mathcal{H}_i)$ are inequivalent irreducible representations. The intertwiners from $\pi$ to itself then have the form $T=\bigoplus_i I \otimes T_i$ for some $T_i\in B(\mathcal{K}_i)$.

We work with representations of the unitary group on a finite-dimensional Hilbert space $\mathcal{H}$. Fix a basis $\{\vert i \rangle \vert i=1,\ldots,d\}$ of $\mathcal{H}$. The \textit{trivial representation} is the mapping $\mathcal{U}(\mathcal{H})\rightarrow S^1$, $U\mapsto 1$; the \textit{fundamental representation} is the identity mapping $\mathcal{U}(\mathcal{H})\rightarrow \mathcal{U}(\mathcal{H})$; and the \textit{contragredient representation} is the mapping $\mathcal{U}(\mathcal{H})\rightarrow \mathcal{U}(\mathcal{H})$, $U\mapsto\bar{U}$, where the complex conjugate is with respect to the fixed basis. These are all irreducible representations, and inequivalent for $d>2$. 

\begin{lemma} [\cite{Mel24}] \label{lem: uni-irrep}
    Consider the representation $\pi:\mathcal{U}(\mathcal{H})\rightarrow \mathcal{U}(\mathcal{H}\otimes\mathcal{H})$, $U\mapsto U\otimes\bar{U}$. Then, any intertwiner $T$ of $\pi$ can be expressed as $T=\alpha|\phi^+\rangle\!\langle \phi^+|+\beta\Pi$,
    where $\Pi={I}-|\phi^+\rangle\!\langle \phi^+|$ is the orthogonal projector onto $\mathcal{K}=|\phi^+\rangle^\perp$, where $|\phi^+\rangle\in\mathcal{H}\otimes\mathcal{H}$ is the maximally entangled state.\footnote{This is equivalent to the decomposition as the direct sum of two irreducible representations: the trivial representation on the subspace $\mathrm{span}\{|\phi^+\rangle\}$ and an irreducible representation on $\mathcal{K}=|\phi^+\rangle^\perp$.} In particular, by orthogonality of the projectors,
    \begin{align}
        \int(U\otimes\bar{U})T(U\otimes\bar{U})^\dag dU=\frac{\mathrm{Tr}(\Pi T)}{d^2-1}\Pi+\langle{\phi^+}|{T}|{\phi^+}\rangle|\phi^+\rangle\!\langle \phi^+|.
    \end{align}
\end{lemma}
For more details on the representation theory of the unitary group, see for example~\cite{Mel24}.

\subsection{Uncloneable encryption} \label{sec:UE}

The goal of an uncloneable encryption scheme is to encode a classical message as a quantum ciphertext in order to guarantee that two non-interacting adversaries cannot both learn the message, even when given the encryption key. This is a security notion that is impossible classically because any classical ciphertext can be copied. We formally define such a primitive and its security in this section. 

\begin{definition}[QECM] \label{def:qecm}
    A \textit{quantum encryption of classical messages (QECM)} is given by a tuple $\texttt{Q}=(K,X,A,\mu,\{\sigma^k_x\}_{k\in K,x\in X})$, where
    \begin{itemize}
        \item $K$ is a set, representing the encryption keys;
        \item $X$ is a finite set, representing the messages;
        \item $A$ is a register, representing the system holding the encrypted messages;
        \item $\mu$ is a probability measure on $K$, representing the key distribution;
        \item $\sigma^k_x\in D(\mathcal{H}_A)$ is a quantum state, representing the encryption of message $x$ with key $k$.
    \end{itemize}
\end{definition}
    
\begin{definition}[Correctness] \label{def:corr-qecm}
    We say a QECM is \textit{$\eta$-correct} if there exists a family of CPTP maps $\Phi^k:B(\mathcal{H}_A)\rightarrow B(\mathcal{H}_M)$, called decryption maps, such that for all $k\in K$ and $x\in M$,
    \begin{align*}
        \langle{x}\vert{\Phi^k(\sigma^k_x)}\vert{x}\rangle\geq\eta.
    \end{align*}
    We say that the QECM is \textit{correct} if it is $1$-correct.
    We say a family of QECMs $\{\texttt{Q}_\lambda\}_{\lambda\in\mathbb{N}}$ is an \textit{efficient QECM} if key sampling, encrypted message preparation, and decryption can be implemented in polynomial time in $\lambda$.
\end{definition}

Note that correctness is equivalent to the orthogonality condition $\text{Tr}(\sigma^k_{x}\sigma^k_{x'})=0$ for $k\in K$ and $x\neq x'\in X$. Throughout this work, the property of correctness is implicitly required, as to perfectly decrypt a QECM one requires states that are orthogonal.

\begin{definition}[Cloning attack] \label{def:cloning-att}
    A \textit{cloning attack} against a QECM $\texttt{Q}=(K,X,A,\mu,\{\sigma^k_x\}_{k\in K,x\in X})$ is a tuple $\texttt{A}=(B,C,\{B^k_x\}_{k\in K,x\in X},\{C^k_x\}_{k\in K,x\in X},\Phi)$, where
    \begin{itemize}
        \item $B$ and $C$ are registers, representing Bob and Charlie's systems, respectively;
        \item $\{B^k_x\}_{x\in X}\subseteq B(\mathcal{H}_B)$ and $\{C^k_x\}_{x\in X}\subseteq B(\mathcal{H}_C)$ are POVMs, representing Bob and Charlie's measurements given key $k$, respectively;
        \item $\Phi:B(\mathcal{H}_A)\rightarrow B(\mathcal{H}_{BC})$ is a CPTP map, representing the cloning channel.
    \end{itemize}
    The \textit{success probability} of $\texttt{A}$ against $\texttt{Q}$ is
    \begin{align}
        \mathfrak{c}(\texttt{Q},\texttt{A})=\int\frac{1}{|X|}\sum_{x\in X}\text{Tr}\left[{(B^k_x\otimes C^k_x)\Phi(\sigma^k_x)}\right]d\mu(k).
    \end{align}
    The \textit{cloning value} of $\texttt{Q}$ is $\mathfrak{c}(\texttt{Q})=\sup_{\texttt{A}}\mathfrak{c}(\texttt{Q},\texttt{A})$, where the supremum is over all cloning attacks. We say a QECM is \textit{$\delta$-uncloneable secure} if $\mathfrak{c}(\texttt{Q})\leq\frac{1}{|X|}+\delta$.

    For a function $f:\mathbb{N}\rightarrow[0,1]$, we say a family of QECMs $\{\texttt{Q}_\lambda\}$ is \textit{$f$-uncloneable secure} if $\texttt{Q}_\lambda$ is $f(\lambda)$-uncloneable secure for all $\lambda$. We additionally say $\{\texttt{Q}_\lambda\}$ is \textit{uncloneable secure} if $\lim\limits_{\lambda\rightarrow\infty}f(\lambda)=0$; and $\{\texttt{Q}_\lambda\}$ is \textit{strongly uncloneable secure} if $f$ is a negligible function.
\end{definition}

\begin{definition}[Cloning-distinguishing attack] \label{def:cloning-distinguishing-att}
    A \textit{cloning-distinguishing attack} against a QECM $\texttt{Q}=(K,X,A,\mu,\{\sigma^k_x\}_{k\in K,x\in X})$ is a tuple $\texttt{A}=(\{x_0,x_1\},B,C,\{B^k_b\}_{k\in K,b\in\{0,1\}},\{C^k_b\}_{k\in K,b\in\{0,1\}},\Phi)$, where
    \begin{itemize}
        \item $x_0\neq x_1\in X$ are distinct messages, representing the two messages to be distinguished;
        \item $B$ and $C$ are registers, representing Bob and Charlie's systems, respectively;
        \item $\{B^k_b\}_{b\in\{0,1\}}\subseteq B(\mathcal{H}_B)$ and $\{C^k_b\}_{b\in\{0,1\}}\subseteq B(\mathcal{H}_C)$ are POVMs, representing Bob and Charlie's measurements given key $k$, respectively;
        \item $\Phi:B(\mathcal{H}_A)\rightarrow B(\mathcal{H}_{BC})$ is a CPTP map, representing the cloning channel.
    \end{itemize}
    The \textit{success probability} of $\texttt{A}$ against $\texttt{Q}$ is
    \begin{align}
        \mathfrak{cd}(\texttt{Q},\texttt{A})=\int\frac{1}{2}\sum_{b\in\{0,1\}}\text{Tr}\left[{(B^k_b\otimes C^k_b)\Phi(\sigma^k_{x_b})}\right]d\mu(k).
    \end{align}
    The \textit{cloning-distinguishing value} of $\texttt{Q}$ is $\mathfrak{cd}(\texttt{Q})=\sup_{\texttt{A}}\mathfrak{cd}(\texttt{Q},\texttt{A})$, where the supremum is over all cloning-distinguishing attacks. We say a QECM is \textit{$\delta$-uncloneable-indistinguishable secure} if~$\mathfrak{cd}(\texttt{Q})\leq\frac{1}{2}+\delta$.

    For a function $f:\mathbb{N}\rightarrow[0,1]$, we say a family of QECMs $\{\texttt{Q}_\lambda\}$ is \textit{$f$-uncloneable-in\-dis\-ting\-ui\-sha\-ble secure} if $\texttt{Q}_\lambda$ is $f(\lambda)$-uncloneable-indistiguishable secure for all $\lambda$. We additionally say~$\{\texttt{Q}_\lambda\}$ is \textit{uncloneable-indistinguishable secure} if $\lim\limits_{\lambda\rightarrow\infty}f(\lambda)=0$; and $\{\texttt{Q}_\lambda\}$ is \textit{strongly uncloneable-indistinguishable secure} if $f$ is a negligible function.
\end{definition}

Note that if $X=\{0,1\}$, then the notions of uncloneable security and uncloneable-in\-dis\-ting\-ui\-sha\-ble security are equivalent. In general, uncloneable-indistinguishable security implies uncloneable security~\cite{BL20}. Also due to \cite{BL20}, uncloneable-indistinguishable security implies the indistinguishable security, a standard cryptographic notion. Further, due to~\cite{HKNY24}, an uncloneable-indistinguishable secure QECM with a one-bit message can be used to construct an uncloneable-indistinguishable secure QECM with arbitrary message size, under the assumption of quantum polynomial-time adversaries, and a primitive called decomposable quantum randomised encoding, which follows from the existence of one-way functions~[\hyperref[ref:BY22]{42}]. Hence, we concentrate on uncloneable security for QECMs with one-bit messages.

The above describes uncloneable encryption in the prepare-and-measure picture via QECMs. A dual description known as the entanglement-based picture also exists where uncloneable encryption admits a characterisation via monogamy-of-entanglement games.

\begin{definition}[Monogamy-of-entanglement game] \label{def:moe-game}
    A \textit{monogamy-of-entanglement (MoE) game} is a tuple $\texttt{G}=(\Theta,X,A,\mu,\{A^\theta_x\}_{\theta\in\Theta,x\in X})$, where
    \begin{itemize}
        \item $\Theta$ is a set, representing the questions;
        \item $X$ is a finite set, representing the answers;
        \item $A$ is a register, representing Alice's system;
        \item $\mu$ is a probability measure on $\Theta$, representing the question distribution.
        \item $\{A^\theta_x\}_{x\in X}\subseteq B(\mathcal{H}_A)$ is a POVM, representing Alice's measurements given question $\theta$.
    \end{itemize}
    A \textit{strategy} for an MoE game $\texttt{G}$ is a tuple $\texttt{S}=(B,C,\{B^\theta_x\}_{\theta\in\Theta,x\in X},\{C^\theta_x\}_{\theta\in \Theta,x\in X},\rho_{ABC})$, where
    \begin{itemize}
        \item $B$ and $C$ are registers, representing Bob and Charlie's systems, respectively;
        \item $\{B^\theta_x\}_{x\in X}\subseteq B(\mathcal{H}_B)$ and $\{C^\theta_x\}_{x\in X}\subseteq B(\mathcal{H}_C)$ are POVMs, representing Bob and Charlie's measurements given question $\theta$, respectively;
        \item $\rho_{ABC}\in D(\mathcal{H}_{ABC})$ is the shared quantum state.
    \end{itemize}
    The \textit{winning probability} of $\texttt{S}$ at $\texttt{G}$ is
    \begin{align}
        \mathfrak{w}(\texttt{G},\texttt{S})=\int\sum_{x\in X}\text{Tr}\left[{(A^\theta_x\otimes B^\theta_x\otimes C^\theta_x)\rho_{ABC}}\right]d\mu(\theta).
    \end{align}
    The \textit{quantum value} of $\texttt{G}$ is $\mathfrak{w}(\texttt{G})=\sup_{\texttt{S}}\mathfrak{w}(\texttt{G},\texttt{S})$, where the supremum is over all strategies.
\end{definition}

\section*{References}

\begin{itemize}
\item[ ]
\begin{itemize}\itemsep1em
    \item[{[32]}]\label{ref:TCR09} Tomamichel, M., Colbeck, R., Renner, R.: A fully quantum asymptotic equipartition property. IEEE Transactions on Information Theory \textbf{55}(12), 5840–5847 (2009) \href{https://doi.org/10.1109/TIT.2009.2032797}{https://doi.org/10.1109/TIT.2009.2032797}

    \item[{[33]}]\label{ref:KRS09} K\"onig, R., Renner, R., Schaffner, C.: The operational meaning of min- and max-entropy. IEEE Transactions on Information Theory \textbf{55}(9), 4337–4347 (2009) \href{https://doi.org/10.1109/TIT.2009.2025545}{https://doi.org/10.1109/TIT.2009.2025545}

    \item[{[34]}]\label{ref:DBWR14} Dupuis, F., Berta, M., Wullschleger, J., Renner, R.: One-shot decoupling. Communications in Mathematical Physics \textbf{328}, 251–284 (2014) \href{https://doi.org/10.1007/s00220-014-1990-4}{https://doi.org/10.1007/s00220-014-1990-4}

    \item[{[35]}]\label{ref:HOW07} Horodecki, M., Oppenheim, J., Winter, A.: Quantum state merging and negative information. Communications in Mathematical Physics \textbf{269}(1), 107–136 (2007) \href{https://doi.org/10.1007/s00220-006-0118-x}{https://doi.org/10.1007/s00220-006-0118-x}

    \item[{[36]}]\label{ref:Ber09} Berta, M.: Single-shot quantum state merging. arXiv preprint arXiv:0912.4495 (2009) \href{https://doi.org/10.48550/arXiv.0912.4495}{https://doi.org/10.48550/arXiv.0912.4495}

    \item[{[37]}]\label{ref:DLT02} DiVincenzo, D.P., Leung, D.W., Terhal, B.M.: Quantum data hiding. IEEE Transactions on Information Theory \textbf{48}(3), 580–598 (2002) \href{https://doi.org/10.1109/18.985948}{https://doi.org/10.1109/18.985948}

    \item[{[38]}]\label{ref:DCEL09} Dankert, C., Cleve, R., Emerson, J., Livine, E.: Exact and approximate unitary 2-designs and their application to fidelity estimation. Physical Review A \textbf{80}(1), 012304 (2009) \href{https://doi.org/10.1103/PhysRevA.80.012304}{https://doi.org/10.1103/PhysRevA.80.012304}

    \item[{[39]}]\label{ref:CLLW16} Cleve, R., Leung, D., Liu, L., Wang, C.: Near-linear constructions of exact unitary 2-designs. Quantum Information \& Computation \textbf{16}(9-10), 721–756 (2016) \href{https://doi.org/10.26421/QIC16.9-10-1}{https://doi.org/10.26421/QIC16.9-10-1}

    \item[{[40]}]\label{ref:SDTR13} Szehr, O., Dupuis, F., Tomamichel, M., Renner, R.: Decoupling with unitary approximate two-designs. New Journal of Physics \textbf{15}(5), 053022 (2013) \href{https://doi.org/10.1088/1367-2630/15/5/053022}{https://doi.org/10.1088/1367-2630/15/5/053022}

    \item[{[41]}]\label{ref:VY16} Vidick, T., Yuen, H.: A simple proof of Renner’s exponential De Finetti theorem. arXiv preprint arXiv:1608.04814 (2016) \href{https://doi.org/10.48550/arXiv.1608.04814}{https://doi.org/10.48550/arXiv.1608.04814}

    \item[{[42]}]\label{ref:BY22} Brakerski, Z., Yuen, H.: Quantum garbled circuits. In: Proceedings of the 54th Annual ACM SIGACT Symposium on Theory of Computing, pp. 804–817 (2022). \href{https://doi.org/10.1145/3519935.3520073}{https://doi.org/10.1145/3519935.3520073}
\end{itemize}
\end{itemize}

\end{document}